\documentclass[journal]{IEEEtran}
\usepackage{mathrsfs}
\usepackage{pifont}
\usepackage{bbding}
\usepackage{amsmath}
\usepackage{multirow}
\usepackage{url}
\usepackage{comment}
\usepackage{stfloats}
\usepackage{amssymb}
\usepackage{algorithm}
\usepackage{algpseudocode}

\usepackage{makecell}
\usepackage{booktabs}
\usepackage{graphicx} 
\usepackage{epstopdf}
\usepackage{amsthm}
\usepackage{subfigure,balance}
\usepackage{amsfonts, epsfig, latexsym}
\usepackage{bm,color}
\usepackage{caption}
\usepackage{cases}
\usepackage[noadjust]{cite}
\usepackage{subcaption}
\usepackage{array}

\usepackage{verbatim,empheq}
\IEEEoverridecommandlockouts

\newtheorem{theorem}{Theorem}

\begin{document}
	
	\allowdisplaybreaks 
	
	\title{Terahertz Beam Squint Mitigation via Six-Dimensional Movable Antennas}
	\author{Yike Xie, \IEEEmembership{Student Member, IEEE}, Weidong Mei, \IEEEmembership{Member, IEEE}, Dong Wang, \IEEEmembership{Student Member, IEEE}, Yingqi Wen, \IEEEmembership{Student Member, IEEE}, Zhi Chen, \IEEEmembership{Senior Member, IEEE}, Jun Fang, \IEEEmembership{Senior Member, IEEE}, Wei Guo, \IEEEmembership{Member, IEEE}, Boyu Ning, \IEEEmembership{Member, IEEE}

		\thanks{Part of this work has been presented at the 2025 IEEE International Conference on Communications, Montreal, Canada \cite{xie}.
			
		Yike Xie, Weidong Mei, Dong Wang, Zhi Chen, Jun Fang, Wei Guo, and Boyu Ning are with the National Key Laboratory of Wireless Communications, University of Electronic Science and Technology of China, Chengdu 611731, China (e-mails: ykxie@std.uestc.edu.cn; wmei@uestc.edu.cn; DongwangUESTC@outlook.com; chenzhi@uestc.edu.cn; junfang@uestc.edu.cn; guowei@uestc.edu.cn; boydning@outlook.com).
		
		Yingqi Wen is with the Glasgow College,  University of Electronic Science and Technology of China, Chengdu 611731, China (e-mail: 2023190504033@std.uestc.edu.cn).}}
	\maketitle  
	
	\begin{abstract}
		Analog beamforming holds great potential for future terahertz (THz) communications due to its ability to generate high-gain directional beams with low-cost phase shifters. However, conventional analog beamforming may suffer substantial performance degradation in wideband systems due to the beam squint effect. Instead of relying on high-cost true-time delayers, we propose an efficient six-dimensional movable antenna (6DMA) architecture to mitigate the beam-squint effect. This design is motivated by the fact that antenna repositioning and rotation can alter the correlation of steering vectors in both spatial and frequency domains. In particular, we study a wideband wide-beam coverage problem in this paper, aiming to maximize the minimum beamforming gain over a given range of azimuth/elevation angles and frequencies by jointly optimizing the analog beamforming vector, the MA positions within a two-dimensional (2D) region, and the three-dimensional (3D) rotation angles of the antenna array. However, this problem is non-convex and intractable to solve optimally due to the coupling of the spatial and frequency domains and that of the antenna weights, positions and rotation. To tackle this problem, we first derive an optimal solution to it in a special case with azimuth or elevation angle coverage only. It is shown that rotating a uniform linear array (ULA) is sufficient to achieve global optimality and eliminate beam-squint effects. While for other general cases, an alternating optimization (AO) algorithm is proposed to obtain a high-quality suboptimal solution, where the antennas' beamforming weights, positions, and rotation angles are alternately optimized by combining successive convex approximation (SCA), sequential update with Gibbs sampling (GS), and hybrid coarse- and fine-grained search. Simulation results demonstrate that our proposed scheme can significantly outperform conventional antenna arrays without antenna movement or rotation, thus offering a cost-effective solution for wideband transmission over THz bands.
	\end{abstract}

	\section{Introduction}
	Terahertz (THz) communication technology, covering the frequency range from $0.1$ THz to $10$ THz, has emerged as a promising solution to the spectrum congestion faced by today’s fifth-generation (5G) wireless systems \cite{chen1,ning1}. Its ultra-wide bandwidth is anticipated to significantly boost data rates and reduce latency beyond what millimeter wave (mmWave) communications can offer, thereby enabling a variety of cutting-edge applications such as extended reality (XR) \cite{XR}, smart healthcare, and vehicle-to-everything (V2X) communications \cite{X2T}. However, THz systems also face a pronounced beam-squint issue due to their large bandwidth. In hybrid precoding architectures, an analog beamformer can steer a directional beam toward the desired physical direction and achieve the full array gain at the center frequency. However, as the bandwidth increases, the beam direction becomes frequency-dependent because the phase shifters apply frequency-independent phase shifts. Consequently, the beam can deviate from the intended direction at off-center frequencies, resulting in a significant loss of array gain.
    
	To mitigate beam squint, several works have proposed employing true-time-delay (TTD) elements to realize frequency-dependent phase shifts \cite{dai1,tan_delay,tan_quasi}. In particular, the delay–phase precoding (DPP) introduces a time-delay network that comprises a small number of delay elements between the radio-frequency (RF) chains and the phase shifters in the hybrid precoding architecture. By jointly controlling delay and phase, DPP can effectively reduce array-gain loss across the wideband spectrum. However, the high cost of THz TTD components and the required modifications to standard phased-array hardware make its large-scale deployment challenging in practice.
    
    Recently, movable antennas (MA) technology has emerged as a new technology to improve wireless communication performance, which enables local antenna movement within a given region at the transmitter (Tx) and/or receiver (Rx) \cite{ma3,ma-tut,ref3,Ning_Mag}. Compared with conventional fixed-position antennas (FPAs), MAs can leverage this new spatial degree of freedom (DoF) to reshape wireless channels and alter the correlation among steering vectors in favor of wireless transmission. As a result, MA systems can attain a similar or even superior performance to FPAs with a much smaller number of RF chains and phase shifters, thereby reducing hardware costs and energy consumption \cite{ma-tut}. Inspired by their promising benefits, the performance of MA systems has been extensively investigated under various system setups, e.g., flexible beamforming \cite{dong,zhu21,ref2,mamulti}, intelligent reflecting surface-aided wireless communications \cite{wei_Friend,gao_ma}, secure communications \cite{mei1,cheng1}, point-to-point communications \cite{ma1,ma22,MaH1}, multi-user systems \cite{zhu1,wu1}, antenna trajectory design \cite{li_Trajectory}, cognitive radio \cite{weix1}, covert communications \cite{hu_cov}, etc. A general optimization framework has been developed to address the inherent non-convexity of MA position optimization in \cite{Liu_General} via discrete sampling. Beyond antenna repositioning, antenna rotation has also been exploited to improve wireless communication and sensing performance by offering additional DoFs. This thus gives rise to a new six-dimensional MA (6DMA) architecture \cite{6DMA}, which exploits both three-dimensional (3D) antenna repositioning and 3D antenna rotation at the same time \cite{6DMA,shao1,shao2,shao3,zheng1,hua1}. Notably, prior studies \cite{wenyq} have shown that antenna rotation can reshape the spatial correlation among array responses associated with different angles, thereby enabling more flexible beamforming. It is also well known that beam-squint effects are less (more) pronounced when the user direction is closer to (farther from) the array boresight \cite{dai1}. Motivated by these observations, the earlier version of this work \cite{xie} exploited antenna rotation to mitigate beam-squint effects.
	\begin{figure}[t]
        \centering
		\subfigure[GCS]{\includegraphics[width=0.3\textwidth]{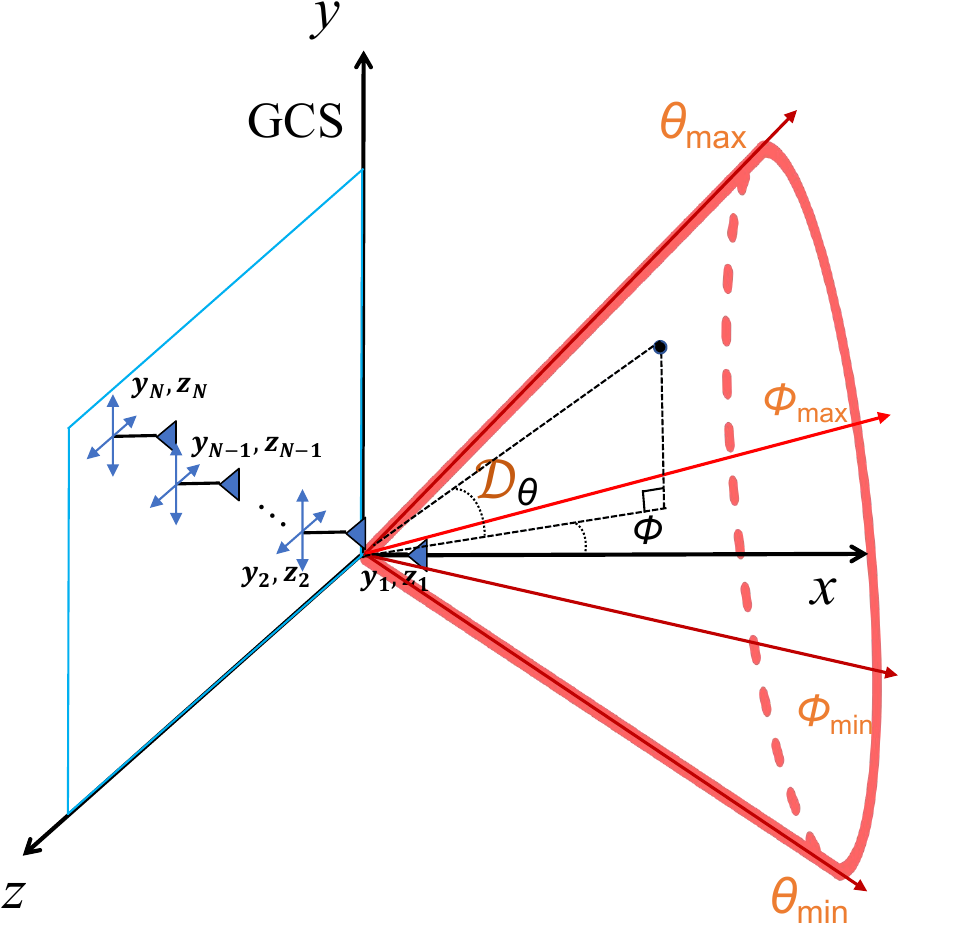}}\label{h_upa1}
		\subfigure[LCS]{\includegraphics[width=0.3\textwidth]{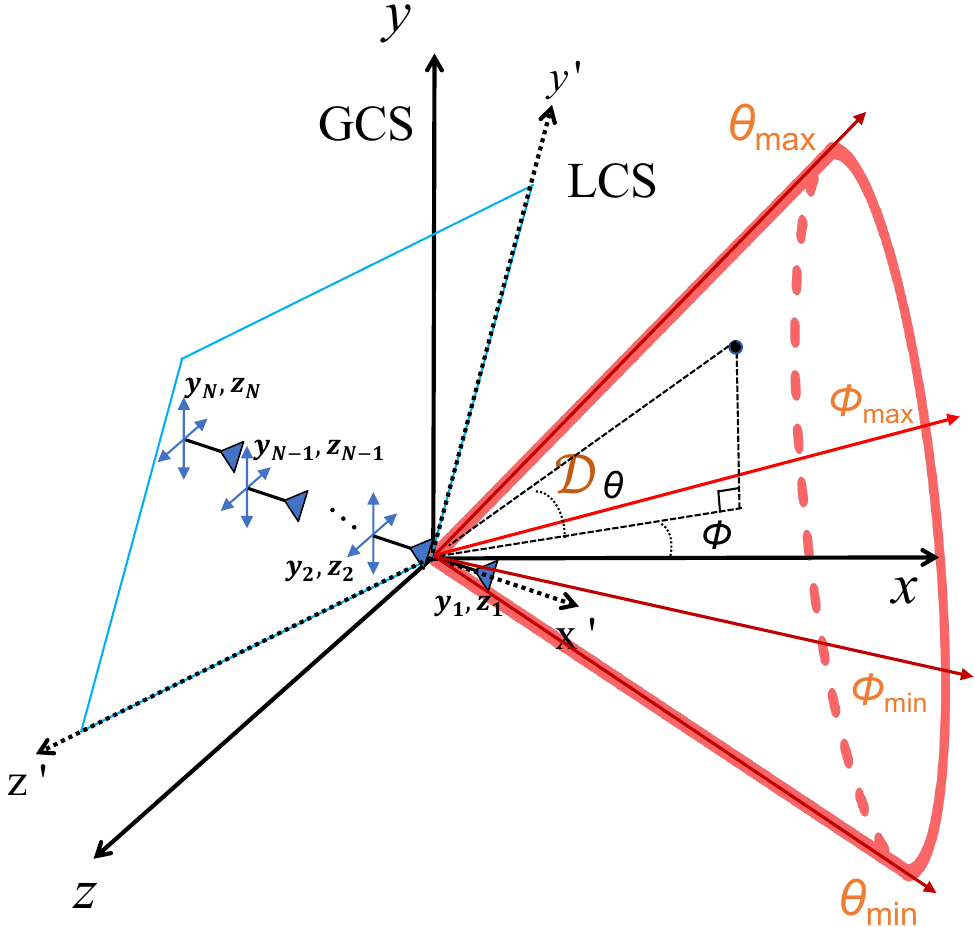}}\label{h_upa2}
		\caption{6DMA-enhanced wideband wide-beam coverage.}	
		\label{fig1}
        \vspace{-6pt}
	\end{figure}
    
	To jointly exploit the benefits of antenna repositioning and rotation, this paper investigates a new 6DMA-based approach to mitigate beam squint and enable joint wide-beam and wideband coverage in the THz band, without relying on TTDs. The main contributions of this paper are summarized below.
    \begin{itemize}
        \item We formulate a new wideband wide-beam coverage problem in this paper, aiming to maximize the minimum beamforming gain over a two-dimensional (2D) angular region and given frequency range by jointly optimizing the transmit beamforming, the antenna positions within a 2D movement region, and the 3D rotation angles of the antenna array, as shown in Fig.\,\ref{fig1}. However, the resulting problem is non-convex and generally difficult to optimally solve due to the intricate coupling of the spatial and frequency domains. To gain insights, we first consider a special case with a one-dimensional (1D) angular coverage region only and derive a closed-form optimal solution. It is shown that a 1D rotatable uniform linear array (ULA) suffices to achieve global optimality by orienting the array to be perpendicular to the 1D angular region. In contrast, a planar array cannot achieve global optimality, since its rotation cannot simultaneously keep both planar dimensions orthogonal to the 1D angular region.
        \item For other general cases, we propose an alternating optimization (AO) algorithm to obtain a high-quality suboptimal solution by alternately optimizing the MAs' beamforming weights, positions, and rotation angles until convergence is reached. For the beamforming and position optimization subproblems, the successive convex approximation (SCA) algorithm is employed, where the nonconvex objective and constraints are linearized via first-order approximations. For the 3D rotation angle optimization, since the resulting objective is highly involved in the angle variables, we propose a hybrid coarse- and fine-grained searching strategy, followed by a Gibbs sampling (GS) procedure to explore the solution space and reduce the risk of convergence to unfavorable local optima. Numerical results validate our analytical results and also demonstrate that our proposed scheme can achieve much better performance than conventional FPAs without antenna movement or rotation. It is also shown that antenna rotation may play a more significant role than antenna repositioning for mitigating beam-squint effects.
    \end{itemize}

    It is worth noting that MAs have been used in prior works \cite{zhu-bs,star-ris} to mitigate beam-squint effects by exploiting only their continuous position flexibility. In contrast, this paper further reveals the potential of antenna rotation, which can provide even greater benefits than antenna repositioning alone. Moreover, this paper considers a more general system model with 2D antenna movement and joint wideband, wide-angle beam coverage. 
   
	The rest of this paper is organized as follows. Section II presents the system model and problem formulation. Section III presents the special case of 1D angular coverage. Section IV presents the proposed algorithm to solve the formulated optimization problem in other general cases. Lastly, Section V presents numerical results, and Section VI concludes the paper. 
    
	{\it Notations:} Bold symbols in capital letters and small letters denote matrices and vectors, respectively. For a matrix ${ \mathbf{W}}$, ${{ \mathbf{W}}^{ T}}$, ${{ \mathbf{W}}^{H}}$ and ${{ \mathbf{W}}^\dag}$ denote its transpose, conjugate transpose, and conjugate, respectively. The symbols $\left|  \cdot   \right|$ and $\angle ( \cdot )$ denote the modulus and the angle of a complex number, respectively. The symbol $\left\|   \cdot  \right\|$ denotes the Euclidean norm of a vector. Moreover, $\text{diag}({ \mathbf{w}})$ denotes a diagonal matrix with the elements of ${ \mathbf{w}}$ on its diagonal. The sub-gradient of $f$ at $\mathbf{A}$ is denoted by ${\partial _{ \mathbf{A}}}f$. For a complex number $s$, $s \sim \mathcal{CN}(0,\sigma^2)$ means that it is a circularly symmetric complex Gaussian (CSCG) random variable with zero mean and variance $\sigma^2$. 

\begingroup
\allowdisplaybreaks
\section{System Model and Problem Formulation}
\subsection{System Model}
    As shown in Fig.\,\ref{fig1}, we consider a wideband THz communication system, where the transmitter is equipped with a 6DMA array comprising $N$ antennas. Each antenna can flexibly move within a confined 2D region of size $A \times A$, denoted as ${\cal C}_t$, and the entire 2D array can be rotated in 3D space. Let $B$ and $f_c$ denote the total bandwidth and the carrier frequency of the system, respectively. This paper focuses on a wide-beam coverage problem, aiming to achieve a uniform beam gain over all directions within a given region in the angular domain (i.e., $\cal D$ in Fig. \ref{fig1}). Given the global coordinate system (GCS) established in Fig.\,\ref{fig1}(a), denote by $\theta_{\min}(\theta_{\max})$ and $\phi_{\min}(\phi_{\max})$ the minimum (maximum) elevation and azimuth angles of departure (AoDs) for the coverage region, respectively, with $\theta_{\min}<\theta_{\max}$ and  $\phi_{\min}<\phi_{\max}$. Thus, we have $\cal D = [\theta_{\min},\theta_{\max}] \times [\phi_{\min},\phi_{\max}]$, where ``$\times$'' denotes the Cartesian product. To ease practical implementation, we assume that analog beamforming is adopted at the transmitter.

        \begin{figure}[t]
		\centerline{\includegraphics[width=0.35\textwidth]{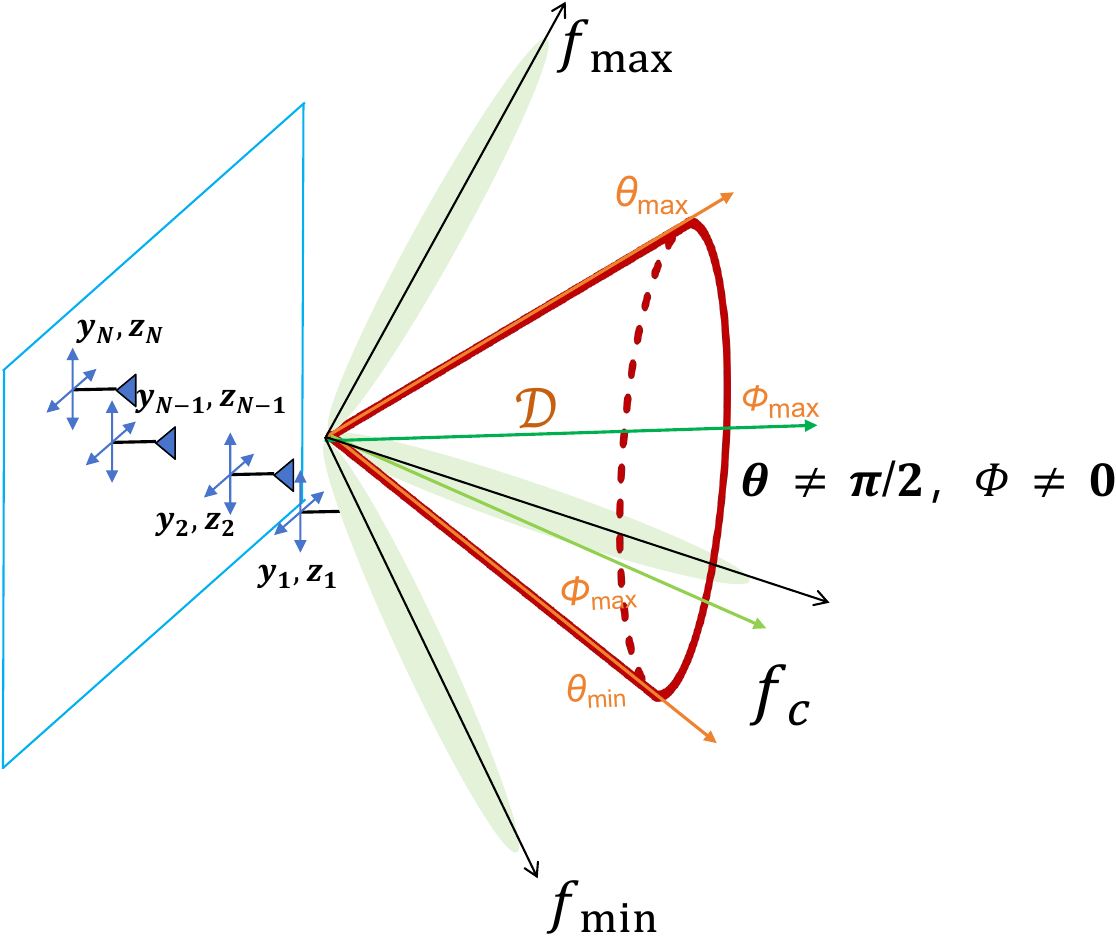}}
		\caption{Illustration of beam squint with different AoDs.}
		\label{2D_bs}
        \vspace{-6pt}
	\end{figure}
	However, unlike narrowband systems with the bandwidth much smaller than the carrier frequency (i.e., $B \ll f_c$), wideband transmission in our considered system introduces non-negligible beam-squint effects, causing the beam gain to vary across the frequency band for any given direction in $\cal D$, as illustrated in Fig. \ref{2D_bs}. Consequently, applying conventional wide-beam designs for narrowband systems (see \cite{ning1} and \cite{dong}) can result in significant performance degradation. To address this issue, and motivated by recent advances in 6DMAs for flexible beamforming, we propose a new approach for wideband wide-beam coverage that leverages the 2D movement and 3D rotation of multiple antennas in a 6DMA array. To characterize the rotation of the antenna array and the position of each antenna, we establish a local coordinate system (LCS) for any given rotational angles of the 6DMA array, assuming that the rotated array is parallel to the $y-z$ plane in the LCS, as shown in Fig.\,\ref{fig1}(b). Then, the coordinate of its $n$-th antenna in the LCS can be expressed as
	\begin{equation}
		\mathbf{p}_n=\left[0,y_n,z_n\right]^T,
	\end{equation}
    where $y_n$ and $z_n$ represent the coordinates of the $n$-th antenna along the $y$-axis and $z$-axis, respectively. Let $\mathbf{p} = [\mathbf{p}_1,\mathbf{p}_2,...,\mathbf{p}_N] \in {\mathbb R}^{2 \times N}$ denote the APV of the MA array and $\mathbf{r}=\{\alpha,\beta,\gamma\}$ denote the rotational angle vector of the planar MA array from the GCS to the LCS, as shown in Fig. \ref{fig1}, where $\alpha \in [0,2\pi]$, $\beta \in [0,2\pi]$ and $\gamma \in [0,2\pi]$ represent the rotational angles around the $x$-, $y$- and $z$-axes, respectively.
	The relationship between the GCS and LCS can be characterized by three rotation matrices corresponding to $x$-, $y$-, and $z$-axes, i.e.,
	\begin{align}
		\mathbf{R}_x(\alpha ) &= \left[ {\begin{array}{*{20}{c}}
				1&0&0\\
				0&{{c_\alpha }}&{ - {s_\alpha }}\\
				0&{{s_\alpha }}&{{c_\alpha }}
		\end{array}} \right],\\
		\mathbf{R}_y(\beta ) &= \left[ {\begin{array}{*{20}{c}}
				{{c_\beta }}&0&{{s_\beta }}\\
				0&1&0\\
				{ - {s_\beta }}&0&{{c_\beta }}
		\end{array}} \right],	\\
		\mathbf{R}_z(\gamma ) &= \left[ {\begin{array}{*{20}{c}}
				{{c_\gamma }}&{ - {s_\gamma }}&0\\
				{{s_\gamma }}&{{c_\gamma }}&0\\
				0&0&1
		\end{array}} \right],
	\end{align}
    where we have defined ${c_\psi} = \cos (\psi )$ and ${s_\psi } = \sin (\psi )$ for notational simplicity, $\psi  \in \left\{ {\alpha ,\beta ,\gamma } \right\}$. Accordingly, the overall rotation matrix can be expressed as the product of the above three rotation matrices, i.e.,
	\begin{align}
		\mathbf{R} &= \mathbf{R}_{x}(\alpha) \mathbf{R}_{y}(\beta) \mathbf{R}_{z}(\gamma) \nonumber\\
		&= \begin{bmatrix}
			c_{\beta} c_{\gamma} & -c_{\beta} s_{\gamma} & s_{\beta} \\
			c_{\alpha} s_{\gamma} + s_{\alpha} s_{\beta} c_{\gamma}  & c_{\alpha} c_{\gamma} - s_{\alpha} s_{\beta} s_{\gamma} & -s_{\alpha} c_{\beta} \\
			s_{\alpha} s_{\gamma} - c_{\alpha} s_{\beta} c_{\gamma} & s_{\alpha} c_{\gamma} + c_{\alpha} s_{\beta} s_{\gamma} & c_{\alpha} c_{\beta}
		\end{bmatrix}.\label{Rmatrix}
	\end{align}
	Based on the above, the position of the $n$-th antenna of the planar MA array in the GCS can be determined as 
	\begin{align}
			{\mathbf{{k}}(n)} &={\mathbf{R}} {\mathbf{p}}_n \nonumber\\
			&={[-{c_\beta }{s_\gamma }  {y_n},( {c_\alpha }{c_\gamma } - {s_\alpha }{s_\beta }{s_\gamma } ) {y_n},( {s_\alpha }{c_\gamma } + {c_\alpha }{s_\beta }{s_\gamma } ) {y_n}]^T}\nonumber\\
			&\qquad+{[ {s_\beta } {z_n},-{s_\alpha }{c_\beta } {z_n},{c_\alpha }{c_\beta }  {z_n}]^T} \nonumber\\
			&=\mathbf{s}_1y(n) + \mathbf{s}_2z(n), \label{k(n,m)}
	\end{align}
	where $\mathbf{s}_2 = {[-{c_\beta }{s_\gamma },( {c_\alpha }{c_\gamma } - {s_\alpha }{s_\beta }{s_\gamma } ),( {s_\alpha }{c_\gamma } + {c_\alpha }{s_\beta }{s_\gamma } )]^T}$ and $\mathbf{s}_2 = {[  {s_\beta } ,-{s_\alpha }{c_\beta } ,{c_\alpha }{c_\beta } ]^T}$.
    
	For any given elevation and azimuth AoDs ($\theta$, $\phi$) within the considered region $\cal D$, the corresponding array response of the MA can be expressed as
	\begin{align}
		\mathbf{a}(f, \theta, \phi, \mathbf{r}, \mathbf{p}) = [&{e^{j2\pi \mathbf{v}^T({\mathbf{s}_1}{y_1} + {\mathbf{s}_2}{z_1})f/c}},{e^{j2\pi \mathbf{v}^T({\mathbf{s}_1}{y_2} + {\mathbf{s}_2}{z_2})f/c}},\nonumber\\ 
			&\cdots,{e^{j2\pi \mathbf{v}^T({\mathbf{s}_1}{y_N} + {\mathbf{s}_2}{z_N})f/c}}]^T, \label{array response2}
		\end{align}
    where $\mathbf{v} = [\cos \theta \cos \phi ,\cos \theta \sin \phi ,\sin \theta]^T$ denotes the direction vector. Let the transmit beamforming of the MA array be denoted as
	\begin{equation}
		\begin{aligned}
			\boldsymbol{\omega}  = \frac{1}{{\sqrt N }}[{e^{j\varphi_1},...,{e^{j\varphi_N}}}]^T. \label{omega}
		\end{aligned}
	\end{equation}
	where $\varphi_n$ denotes the phase shift of the $n$-th antenna element.
    
	Consequently, the beam gain at the AoDs ($\theta$, $\phi$) and any frequency $f$ can be obtained as
	\begin{equation}
			G(f, \theta, \phi, \boldsymbol{\omega}, \mathbf{r}, \mathbf{p})  = {\left| {{{\boldsymbol{\omega }}^H} \mathbf{a}\left( {f, \theta, \phi, \mathbf{r}, \mathbf{p}} \right) }\right|}^2.
	\end{equation}

    To capture the effects of beam squint for each given AoD $(\theta, \phi)$, we adopt the following wideband beam gain \cite{dong,ning3,wang}
    \begin{equation}
		\begin{split}
			G_0(\theta, \phi,\boldsymbol{\omega}, \mathbf{r}, \mathbf{p}) = 
			\mathop {\min }\limits_{f\in {\cal F}} G(f, \theta, \phi, \boldsymbol{\omega}, \mathbf{r}, \mathbf{p}), \label{wideband}
		\end{split}
	\end{equation}
    where ${\cal F} = [f_{\min},f_{\max}]$, with $f_{\min}$ and $f_{\max}$ denoting the minimum and maximum frequencies of interest, respectively. Notably, a larger value of $G_0(\theta, \phi, \boldsymbol{\omega}$ indicates a more uniform beam-gain distribution across all frequencies, which in turn mitigates the severity of beam-squint effects. In addition, $G_0(\theta, \phi, \boldsymbol{\omega}, \mathbf{r}, \mathbf{p})$ is practically relevant to capture the performance of wideband communications. For example, in 5G new radio (NR) systems, different subcarriers allocated to a single user apply the same modulation and coding scheme. Hence, the user's performance is determined by the minimum beam gain over all subcarriers. 
    
	To ensure the wideband beam gain over all AoDs in $\cal D$, we define the following spatial wideband beam gain as
	\begin{align}
			{G_{\min }}(\boldsymbol{\omega}, \mathbf{r}, \mathbf{p}) &= \mathop {\min }\limits_{(\theta, \phi)\in {\cal D}} G_0(\theta, \phi,\boldsymbol{\omega}, \mathbf{r}, \mathbf{p})\nonumber\\
			&=\mathop {\min }\limits_{(\theta, \phi, f)\in {\cal D} \times {\cal F}} G(f, \theta, \phi, \boldsymbol{\omega}, \mathbf{r}, \mathbf{p}). \label{ula_P1}
	\end{align}
	However, compared to the wideband beam gain in \eqref{wideband}, the spatial wideband beam gain in (\ref{ula_P1}) further involves the angular domain, which renders its analysis and optimization more challenging. \vspace{-6pt}	
	
    \subsection{Problem Formulation}
	In this paper, we aim to jointly optimize the analog transmit beamforming $\boldsymbol{\omega}$, the positions of the antennas $\mathbf{p}$, and the rotational angle vector $\mathbf{r}$ to maximize the spatial wideband beam gain in (\ref{ula_P1}). The associated optimization problem is formulated as
	\begin{subequations}
		\label{P1}
		\begin{align}
			\text{(P1)}&\;\mathop {\max }\limits_{\boldsymbol{\omega}, \mathbf{r},\mathbf{p}}\; {G_{\min }}(\boldsymbol{\omega}, \mathbf{r}, \mathbf{p}) \nonumber\\
			\text{s.t.}&\;\left| {\omega_n} \right| = \frac{1}{{\sqrt {N} }},{\rm{  }}\forall n\in\mathcal{N},\\
			&\;\alpha, \beta, \gamma \in [0,2\pi],\\
			&\;||\mathbf{p}_n-\mathbf{p}_m|| \ge d_{\min}, n\ne m, n,m\in\mathcal{N},\label{spacing1}\\
            &\;\mathbf{p}_n \in {\cal C}_t,\label{spacing2}
		\end{align}
	\end{subequations}
    where $\omega_n$ denotes the $n$-th entry of $\boldsymbol{\omega}$, and $d_{\min}$ denotes the minimum distance between any two MAs for avoiding mutual coupling. It is noted that solving (P1) requires only large-scale user distribution information, i.e., the angular region $\cal D$, which varies much more slowly over time than small-scale fading with a short coherence time. Hence, the antenna positions and rotation angles do not need to be updated frequently in practice. This provides the MAs with sufficient time to move and rotate to their respective optimized positions and orientations. Based on the optimized antenna positions and rotation angles, the transmit beamforming can be re-optimized in actual communications based on instantaneous channel state information (CSI) to accommodate real-time performance metrics, such as the achievable rate \cite{Wan_TVT_Hybrid}, etc. However, (P1) is generally difficult to optimally solve due to the continuous spatial and frequency range, the intricate coupling of $\boldsymbol{\omega}$, $\mathbf{r}$, and $\mathbf{p}$ in (\ref{ula_P1}), as well as the unit-modulus constraints. To gain essential insights into beam-squint mitigation via 6DMA, we first discuss several special cases in the following section.
	
	\section{Special Case of 1D Angular Coverage}\label{special}
	In this section, we consider a special case with 1D angular coverage. Without loss of generality, we assume that the target azimuth angle is fixed as $\phi_0$. Thus, the coverage region reduces to ${\cal D}=\{\theta| \theta_{\min} \le \theta \le \theta_{\max}\}$, and the direction vector is given by $\mathbf{v}(\theta) = [\cos\theta\cos\phi_0, \cos\theta\sin\phi_0, \sin\theta]^T$. The frequency-dependent array response vector can be simplified as $\mathbf{a}(f,\theta,\mathbf{r})=[a_1(f,\theta,\mathbf{r}),\cdots,a_N(f,\theta,\mathbf{r})]^H$, where
	\begin{equation}\label{an}
			{a}_n(f,\theta,\mathbf{r}) = {e^{j\frac{{2\pi f}}{c}\mathbf{k}_n^T\mathbf{v}\left( \theta  \right)}} = {e^{j\frac{{2\pi f}}{c}({\bf{r}}_1^T{\bf{v}}(\theta )){x_n}}}.
	\end{equation}
    Based on the above, we present the following theorem.
	\begin{theorem} 
		For a fixed azimuth AoD $\phi_0$, if the antenna movement region ${\cal C}_t$ is sufficiently large, i.e., $A \ge (N-1)d_{\min}$, an optimal solution to (P1) is given by $\alpha^*=0$, $\beta^*=0$, 
        \begin{equation}\label{gamma}
        \gamma^* = \phi_0 + \frac{\pi}{2} + k\pi, \quad k \in \mathbb{Z},
        \end{equation}
        $\boldsymbol{\omega}^* = \frac{1}{\sqrt{N}}[1, \dots, 1]^T$, and $z_n^*=0, n \in \cal N$. The optimal values of $y_n$ are any feasible solution to (P1) subject to the constraints in \eqref{spacing1} and \eqref{spacing2}.
	\end{theorem}
    \begin{IEEEproof}
    It is evident to see ${G_{\min }}(\boldsymbol{\omega}, \mathbf{r}, \mathbf{p}) \leq N$. Hence, it suffices to show that ${G_{\min }}(\boldsymbol{\omega}, \mathbf{r}, \mathbf{p}) =N$ can be achieved by the solutions presented in Theorem 1. As Theorem 1 implies that a linear array is optimal and the optimal values of $y_n$'s are not unique, for simplicity, we consider a ULA that is aligned along the $x$-axis in the LCS after the rotation. As such, the coordinates of the $n$-th antenna are given by $\mathbf{p}_n = [x_n, 0, 0]^T$ in the LCS and
	\begin{equation}
			\mathbf{k}_n = \mathbf{R}\mathbf{p}_n  = \mathbf{r}_1x_n,
	\end{equation}
    in the GCS, respectively, where $$\mathbf{r}_1=\left[{{\mathop{\rm c}} _\beta }{{\mathop{\rm c}} _\gamma }  ,{s_\alpha }{s_\beta }{c_\gamma } + {c_\alpha }{s_\gamma }  ,{s_\alpha }{s_\gamma } - {{\mathop{\rm c}} _\alpha }{s_\beta }{c_\gamma } \right]$$ is the first column of $\mathbf{R}$. 
		By substituting $\alpha^*$, $\beta^*$, and $\gamma^*$ in Theorem 1 into \eqref{Rmatrix}, the first column of the rotation matrix $\mathbf{R}$ simplifies to $\mathbf{r}_1 = [\cos\gamma, \sin\gamma, 0]^T$. The effective phase projection becomes
	\begin{align}
		\mathbf{r}_1^T \mathbf{v}(\theta) &= \cos\gamma \cos\theta \cos\phi_0 + \sin\gamma \cos\theta \sin\phi_0\nonumber\\
	&= \cos\theta \cos(\gamma - \phi_0).\label{rv}
	\end{align}
	Substituting \eqref{gamma} into \eqref{rv}, we have $\cos(\gamma - \phi_0) = \cos(\pi/2) = 0$. Consequently, the phase term in \eqref{an} becomes zero, such that the array response vector becomes frequency- and angle-independent, i.e., $\mathbf{a}(f, \theta) = [1, 1, \cdots, 1]^T, (\theta, f)\in {\cal D} \times {\cal F}$. By setting $\boldsymbol{\omega} = \frac{1}{\sqrt{N}}[1, \cdots, 1]^T$, the resulting beamforming gain is given by
	\begin{equation}
		G(f, \theta) = \left| \boldsymbol{\omega}^H \mathbf{a}(f, \theta) \right|^2 = N, \quad \forall f,\theta.
		\end{equation}
	Thus, the full array gain $N$ is achieved across the entire bandwidth without any beam-squint loss. This completes the proof.
    \end{IEEEproof}
		\begin{figure}[t]
		\centerline{\includegraphics[width=0.48\textwidth]{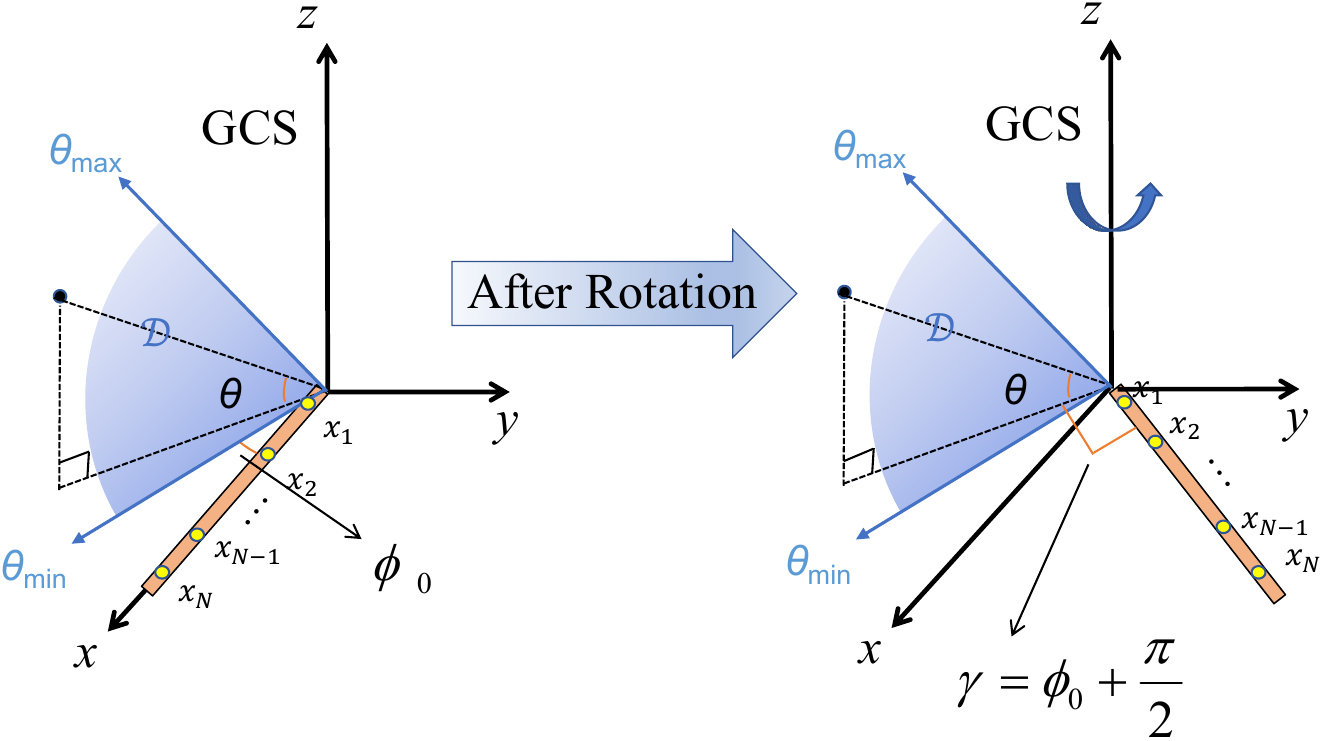}}
		\caption{Illustration of the optimal ULA rotation for 1D angular coverage.}
		\label{1D_coverage}
        \vspace{-6pt}
	\end{figure}
	Theorem 1 implies that, under 1D coverage, ULA rotation provides a perfect solution for eliminating the beam-squint effect. This can be further explained as follows. As shown in Fig. \ref{1D_coverage}, the rotational angles given in Theorem 1 ensure that the ULA is rotated to a direction perpendicular to $\cal D$. As such, each outgoing wave direction between $\theta_{\min}$ and $\theta_{\max}$ is perpendicular to the ULA. Under this configuration, the path difference between any two antenna elements is zero regardless of the elevation angle $\theta$, thereby eliminating beam-squint effects. Nonetheless, in the case of 2D coverage, a perfect zero-beam-squint condition cannot be achieved for ULAs with antenna rotation only. This is expected, as linear arrays can only accommodate the beam squint for either azimuth or elevation AoD within $\cal D$. The rigorous proof is provided in Appendix A.
	In addition, Theorem 1 shows that, somewhat surprisingly, a planar array cannot perfectly eliminate beam squint, even under 1D angular coverage. The underlying reason is that a planar array cannot keep both of its planar dimensions orthogonal to $\cal D$ via rotation. The detailed proof is provided in Appendix B by taking a UPA as an example.

     However, for general 2D angular coverage, the optimal solution to (P1) cannot be easily derived due to the triply coupled domain $(f, \theta, \phi)$. In the next section, we propose an efficient AO algorithm to obtain a suboptimal solution to (P1)

	\section{Proposed Solution to (P1)}
	First, due to the continuous nature of ${\cal D}$ and ${\cal F}$ in (\ref{ula_P1}), we discretize the range of the elevation AoD, azimuth AoD and frequency into $L_1$, $L_2$ and $L_3$ uniformly sampling points. As such, the $l_1$-th sampling point for $\theta$ can be expressed as
	\begin{equation}
		\theta_{l_1}  = \theta _{\min } + \frac{{l_1 - 1}}{{L_1 - 1}}({\theta _{\min }-\theta _{\max }}),{l_1} = 1,2,...,L_1.
	\end{equation}
	The $l_2$-th sampling point for $\phi$ is expressed as
	\begin{equation}
		\phi_{l_2} = \phi _{\min } + \frac{{l_2 - 1}}{{L_2 - 1}}({\phi _{\min }-\phi _{\max }}),{l_2} = 1,2,...,{L_2}.
	\end{equation}
	The $l_3$-th sampling point for $f$ is expressed as
	\begin{equation}
		f_{l_3} = f _{\min } + \frac{{l_3 - 1}}{{L_3 - 1}}({f _{\min }-f _{\max }}),{l_3} = 1,2,...,{L_3}.
	\end{equation}
	Based on the above, the overall spatial-frequency domain in (\ref{ula_P1}), i.e., ${\cal D} \times {\cal F}$, is discretized into a discrete set consisting of $L_1L_2L_3$ sampling points, which is  denoted as ${\cal R} = \{ ({\theta _{{l_1}}},{\phi _{{l_2}}},{f_{{l_3}}})|\forall {l_1},{l_2},{l_3}\}$.\par
	As a result, (P1) can be simplified as
	\begin{subequations}
		\begin{align}
			\text{(P2)}:&\mathop {\max }\limits_{\boldsymbol{\omega}, \mathbf{r} ,\mathbf{p} , \varsigma } \varsigma \\
			\text{s.t. }&G(\vartheta, \boldsymbol{\omega} ,\mathbf{r}, \mathbf{p} ) \ge \varsigma ,\forall \vartheta \in {\cal R} \label{constraint2}\\
			&\left| {\omega_n} \right| = \frac{1}{{\sqrt N }},{\rm{  }}\forall n\in\mathcal{N},\\
			&\alpha, \beta, \gamma \in [0,2\pi],\\
			&||\mathbf{p}_n-\mathbf{p}_m|| \ge d_{\min}, n\ne m, n,m\in\mathcal{N}, \label{mc}\\
            &\;\mathbf{p}_n \in {\cal C}_t,
		\end{align}
	\end{subequations} 
	where $\vartheta=\{\theta,\phi,f\}$ denotes an arbitary element in {$\cal R$}, and $\varsigma$ is an auxiliary variable. However, (P2) remains non-convex due to the coupling of  $\boldsymbol{\omega}$,  $\mathbf{p}$, and  $\mathbf{r}$. Next, we adopt the AO algorithm to tackle it, which alternately optimizes $\boldsymbol{\omega}$, $\mathbf{p}$, and $\mathbf{r}$, with all other variables being fixed.\vspace{-6pt}

	\subsection{Optimizing $\boldsymbol{\omega}$ with Given $\mathbf{r}$ and $\mathbf{p}$}
	First, we optimize the transmit beamforming vector $\boldsymbol{\omega}$ with any given $\mathbf{r}$ and $\mathbf{p}$.
	To deal with the non-convex constraint (\ref{constraint2}), we first rewrite the left-hand side as of (\ref{constraint2})
	\begin{equation}
		\begin{aligned}
			G(\mathbf{\vartheta}, \boldsymbol{\omega}, \mathbf{r}, \mathbf{p}) &= \boldsymbol{\omega}^H \mathbf{a}(\mathbf{\vartheta}, \mathbf{r}, \mathbf{p}) \mathbf{a}(\mathbf{\vartheta}, \mathbf{r}, \mathbf{p})^H \boldsymbol{\omega} \\
			&= \operatorname{Tr}(\mathbf{V}_{\vartheta} \mathbf{W}),
		\end{aligned}
	\end{equation}
	where $\mathbf{V}_{\vartheta}=\mathbf{a}({{\mathbf{\vartheta}}},\mathbf{r},\mathbf{p})\mathbf{a}{(\mathbf{\vartheta},\mathbf{r},\mathbf{p})^H}$ and ${ \mathbf{W}} = {\boldsymbol{\omega}}{\boldsymbol{\omega}^H}$. Then, for any given $\mathbf{r}$ and $\mathbf{p}$, (P2) can be reformulated as 
	\begin{subequations}
		\begin{align}
			\text{(P2.1)}: &\mathop {\max }\limits_{ \mathbf{W} } \varsigma \\
			\text{s.t. }&\text{Tr}(\mathbf{V}_{\vartheta} \mathbf{W}) \ge \varsigma ,\forall \vartheta \in {\cal R},\\
			&{ \mathbf{W}(n,n)} = \frac{1}{ N },\forall n \in N,\\
			&{ \mathbf{W}} \succeq 0,\\
			&\text{rank}(\mathbf{W}) = 1.
		\end{align}
	\end{subequations} \par
	To handle the rank-one constraint $\textrm{rank}(\mathbf{W})=1$, we note that
	\begin{equation}
		\text{rank}({ \mathbf{W}}) = 1 \Leftrightarrow f( \mathbf{W}) \buildrel \Delta \over = {\left\| { \mathbf{W}} \right\| }_ * - {\left\| { \mathbf{W}} \right\|}_2 = 0,
	\end{equation}
	where the nuclear norm ${\left\| { \mathbf{W}} \right\| }_ *$ equals to the summation of all singular value of matrix ${ \mathbf{W}}$, and the spectral norm ${\left\| { \mathbf{W}} \right\|_2}$ equals to the largest singular value of ${ \mathbf{W}}$. Then, the objective function in (P2.1) can be reformulated as
	\begin{equation}
		\mathop {\max }\limits_{ \mathbf{W},\varsigma} \varsigma  - \rho f( \mathbf{W}),
	\end{equation}
	where $\rho > 0 $ is the penalty parameter ensuring that the objective function is small enough if ${\left\| { \mathbf{W}} \right\|_*} - {\left\| { \mathbf{W}} \right\|_2} \ne 0$. However, the objective function is still non-convex due to ${\left\| { \mathbf{W}} \right\|_2}$ represents the maximum singular value of $ \mathbf{W}$. Nonetheless, we could utilize the SCA algorithm to get a locally optimal solution. In particular, for any given local point ${{ \mathbf{W}}^{(i)}}$ in the $i$-th SCA iteration, we replace $f( \mathbf{W})$ as its first-order Taylor expansion, i.e.,
	\begin{equation}
		\begin{aligned}
			f({ \mathbf{W}}) &\ge \widetilde f({ \mathbf{W}}|{{ \mathbf{W}}^{(i)}}){\rm{  }} \buildrel \Delta \over = {| { \mathbf{W}} |_*} - ({|| {{{ \mathbf{W}}^{(i)}}} ||_2}\\ &\qquad+ {\rm{Re}}(\rm{Tr}(({\partial _{{{ \mathbf{W}}^{(i)}}}}{\left\| { \mathbf{W}} \right\|_2})({ \mathbf{W}} - {{ \mathbf{W}}^{(i)}})))),
		\end{aligned}
	\end{equation}
	where the sub-gradient of ${\left\| { \mathbf{W}} \right\|_2}$ can be computed as ${\partial _{{{ \mathbf{W}}^{(i)}}}}{\left\| { \mathbf{W}} \right\|_2} = { \mathbf{s}} {{\mathbf{s}}^H}$, where $\mathbf{s}$ denotes the eigenvector corresponding to the largest eigenvalue of $\mathbf{W}^{(i)}$ \cite{passive}. Based on the above, the optimization problem of $ \mathbf{W}$ in the $i$-th SCA iteration is given by
	\begin{subequations}
		\begin{align}
			\text{(P2.2)}: &\mathop {\max }\limits_{{ \mathbf{W}},{ \varsigma }} \varsigma  - \rho \widetilde f{ \mathbf{(W}|}{{ \mathbf{W}}^{(i)}}{\rm{)}}\\
			\text{s.t. }&\text{Tr}({ \mathbf{V}_{\vartheta} \mathbf{W}}) \ge \varsigma ,\forall \vartheta \in {\cal R},\\
			&{ \mathbf{W}(n,n)} = \frac{1}{ N },\forall n \in N,\\
			&{ \mathbf{W}} \succeq 0.
		\end{align}
	\end{subequations}
	It can be seen that the objective function (P2.2) is currently a linear function of $\varsigma$ and $\mathbf{W}$; thus, (P2.2) can be optimally solved by adopting the interior-point algorithm. After solving (P2.2), we proceed to the $(i+1)$-th SCA iteration for $ \mathbf{W}$ by updating ${{ \mathbf{W}}^{(i + 1)}}$ as the optimal solution to (P2.2). Finally, we perform the singular value decomposition (SVD) on $\mathbf{W}$ as ${ \mathbf{W}} = { \mathbf{U}}_1^H{ \mathbf{\Lambda }}{{ \mathbf{U}}_2}$, where ${ \mathbf{U}}_1$, $ \mathbf{\Lambda }$ and ${{ \mathbf{U}}_2}$ are the left eigenvector matrix, the diagonal matrix of the singular values, and the right eigenvector matrix of $ \mathbf{W}$, respectively. The transmit beamforming solution to (P2.1) can be constructed as $\boldsymbol{\omega}  = 1/\sqrt N {e^{j\arg ({ \mathbf{U}}_1^H   \sqrt { \mathbf{\Lambda }}    \mathbf{q})}}$ accordingly, where $\mathbf{q}$ is the left singular vector corresponding to the largest singular value of $\mathbf{W}$.\vspace{-6pt}
	
	\subsection{Optimizing $\mathbf{p}$ with Given $\boldsymbol{\omega}$ and $\mathbf{r}$}
	Next, we optimize the antenna positions $\mathbf{p}$ with any given analog transmit beamforming $\boldsymbol{\omega}$ and rotation angle vector $\mathbf{r}$ by adopting the SCA. 
	Let $\mathbf{p}^{(i)}=[p_1^{(i)}, p_2^{(i)}, ..., p_N^{(i)}]$ denote the local point of $\mathbf{p}$ in the $i$-th SCA iteration, where $p_n^{(i)}= [0, y^{(i)}_n, z^{(i)}_n ]^T$.
	Firstly, we apply the first-order Taylor expansion to constraint (\ref{mc}), i.e.,
		\begin{align}
			||{\mathbf{p}_n} - {\mathbf{p}_m}||_2^2 &\approx 2[({y_{n,m}} - y_{n,m}^{(i)})y_{n,m}^{(i)} + ({z_{n,m}} - z_{n,m}^{(i)})z_{n,m}^{(i)}]\nonumber\\
			&+ {(y_{n,m}^{(i)})^2} + {(z_{n,m}^{(i)})^2} \ge d_{\min }^2, \label{mc_sca}
		\end{align}
	where ${y_{n,m}} = {y_n} - {y_m}$, ${z_{n,m}} = {z_n} - {z_m}$, ${y^{(i)}_{n,m}} = {y^{(i)}_n} - {y^{(i)}_m}$ and ${z^{(i)}_{n,m}} = {z^{(i)}_n} - {z^{(i)}_m}$, respectively. By replacing constraint (\ref{mc}) with (\ref{mc_sca}), the inter-MA distance constraint is transformed into an affine constraint.
	Second, for any given beamforming $\boldsymbol{\omega}$, we employ the second-Taylor expansion to approximate the left-hand side of (\ref{constraint2}), i.e.,
	\begin{equation}
		\begin{aligned}
			G({\mathbf{\vartheta}},&\boldsymbol{\omega} ,\mathbf{r},\mathbf{p}) \ge f(\mathbf{p}|{\mathbf{p}^{(i)}})\\ 
			&= vec{({\mathbf{p}^T})^T}{\mathbf{A}(\mathbf{\vartheta})}vec({\mathbf{p}^T}) + [{\mathbf{b}(\mathbf{\vartheta})]^T}vec({\mathbf{p}^T}) + {c(\mathbf{\vartheta})},\label{26}
		\end{aligned}
	\end{equation}
	where ${\mathbf{A}(\mathbf{\vartheta})} = [\begin{array}{*{20}{c}}
		{ - \alpha (\mathbf{\vartheta})^2\mathbf{W}}  &{ - {\alpha (\mathbf{\vartheta})}{\beta (\mathbf{\vartheta})}\mathbf{W}}\\
		{ - {\alpha (\mathbf{\vartheta})}{\beta (\mathbf{\vartheta})}\mathbf{W}}&{ - \beta (\mathbf{\vartheta})^2\mathbf{W}}
	\end{array}] \in {\mathbb{R}^{2N \times 2N}}$, $\mathbf{b}(\mathbf{\vartheta})$ and $c(\mathbf{\vartheta})$ are given by (\ref{59}) and (\ref{60}) at the top of next page, respectively,  $u(\mathbf{\vartheta})(\mathbf{\Gamma}^{(i)})={\alpha (\mathbf{\vartheta})}y_{n,m}^{(i)} + {\beta (\mathbf{\vartheta})}z_{n,m}^{(i)} - ({\widehat \phi _n} - {\widehat \phi _m})$, with ${\alpha (\mathbf{\vartheta})} = {\mathbf{\overline \Omega}  (\mathbf{\vartheta})}{\mathbf{R}_1}$, ${\beta (\mathbf{\vartheta})} = {\mathbf{\overline \Omega}  (\mathbf{\vartheta})}{R_2}$, $\mathbf{W} = {\mathbf{I}_N} - \frac{1}{N}{\mathbf{1}_N}$, ${y_{n,m}}^{(i)} = {y_n}^{(i)} - {y_m}^{(i)}$,  ${z_{n,m}}^{(i)} = {z_n}^{(i)} - {z_m}^{(i)}$, and ${\mathbf{\Gamma} ^{(i)}} = \{ y_n^{(i)},y_m^{(i)},z_n^{(i)},z_m^{(i)}\} $. \par
	
	Based on the approximations (\ref{mc_sca}) and (\ref{26}), the optimization problem for updating the antenna positions $\mathbf{p}$ in the $(i+1)$-th iteration can be formulated as:
\begin{subequations}
    \begin{align}
        \text{(P2.3)}: &\mathop {\max }\limits_{{ \mathbf{p}},{ \varsigma }} \varsigma \\
        \text{s.t. }& \text{vec}({\mathbf{p}^T})^T{\mathbf{A}(\mathbf{\vartheta})}\text{vec}({\mathbf{p}^T}) \nonumber \\
        &\quad + [{\mathbf{b}(\mathbf{\vartheta})]^T}\text{vec}({\mathbf{p}^T}) + {c(\mathbf{\vartheta})} \ge \varsigma, \forall \vartheta \in {\cal R},\\
        &||\mathbf{p}_n - \mathbf{p}_m||_2^2 \ge d_{\min}^2, \forall 1 \le n < m \le N, \\
        &\mathbf{p}_n \in \mathcal{C}_l, \forall n \in \mathcal{N},
    \end{align}
\end{subequations}
Problem (P2.3) is a convex quadratically constrained quadratic program (QCQP), which can be efficiently solved by standard convex optimization solvers.\vspace{-6pt}
	
	\begin{figure*}[ht]
		\centering
		\begin{equation}
			\mathbf{b}(\mathbf{\vartheta},n) = 
			\begin{cases} 
				\frac{2}{N} \sum_{m=1}^{N} \left( \alpha(\mathbf{\vartheta})^2 y_{n,m}^{(i)} + \alpha(\mathbf{\vartheta}) \beta(\mathbf{\vartheta}) [z_{n,m}^{(i)} - 2 \alpha(\mathbf{\vartheta}) \sin(u_{l}(\mathbf{\Gamma}^{(i)}))] \right), & 1 \leq n \leq N, \\
				\frac{2}{N} \sum_{m=1}^{N} \left( \beta(\mathbf{\vartheta})^2 y_{n,m}^{(i)} + \alpha(\mathbf{\vartheta}) \beta(\mathbf{\vartheta}) [y_{n,m}^{(i)} - 2 \beta(\mathbf{\vartheta}) \sin(u_{l}(\mathbf{\Gamma}^{(i)}))] \right), & N+1 \leq n \leq 2N, 
			\end{cases}\label{59}
		\end{equation}
		
		\begin{equation}
			c(\mathbf{\vartheta})= \frac{1}{N} \sum_{m=1}^{N} \sum_{n=1}^{N} \left( \cos(u_{l}(\mathbf{\Gamma}^{(i)})) + \sin(u_{l}(\mathbf{\Gamma}^{(i)})) (\alpha(\mathbf{\vartheta}) y_{n,m}^{(i)} + \beta(\mathbf{\vartheta}) z_{n,m}^{(i)}) - \frac{1}{2} (\alpha(\mathbf{\vartheta}) y_{n,m}^{(i)} + \beta(\mathbf{\vartheta}) z_{n,m}^{(i)})^2 \right).\label{60}
		\end{equation}
		\hrule 
	\end{figure*}
	
	\subsection{Optimizing $\mathbf{r}$ with Given $\boldsymbol{\omega}$ and $\mathbf{p}$}
	Finally, we optimize the rotational angle vector $\mathbf{r}$ for any given $\boldsymbol{\omega}$ and $\mathbf{p}$, i.e., 
	\begin{subequations}
		\begin{align}
			\text{(P2.3)}:\quad & \mathop {\max }\limits_{\mathbf{r}} {G_{\min }}(\mathbf{\vartheta},\boldsymbol{\omega},\mathbf{r},\mathbf{p}) \label{eq:p33} \\
			\text{s.t.} & \quad \alpha, \beta, \gamma \in [0,2\pi]. \label{eq:const}
		\end{align}
	\end{subequations}\par
	However, this problem is challenging to solve optimally due to the non-convex objective function w.r.t. $\alpha$, $\beta$, and $\gamma$. In particular, the SCA algorithm may not apply to (P2.3) due to the highly nonlinear expression of (\ref{eq:p33}) in terms of $\alpha$, $\beta$, and $\gamma$. To tackle this issue, we apply a hybrid searching method via discrete sampling. Specifically, assume that the angular intervals $[0,2\pi]$ for $\alpha$, $\beta$ and $\gamma$ are discretized into ${\widetilde N_x}$, ${\widetilde N_y}$ and ${\widetilde N_z}$ segments, respectively. This results in a total of ${\widetilde N_{tot}} = {\widetilde N_x} \times {\widetilde N_y} \times {\widetilde N_z}$ cuboids in their 3D feasible region. Let ${\mathbf{r}_m} = {\left[ {{\widetilde \alpha}_m},{{\widetilde \beta }_m},{{\widetilde \gamma }_m} \right]^T}$ represent the center of the $m$-th cuboid, $1 \le m \le \widetilde N_{tot}$, and the coordinates of each center are given by
	\begin{subequations}
		\begin{align}
			{\widetilde \alpha _m} &= \frac{{\pi (2{n_x} + 1 - {{\widetilde N}_x})}}{{2{{\widetilde N}_x}}},{n_x} = [m/({\widetilde N_y}{\widetilde N_z})],\\
			{\widetilde \beta _m} &= \frac{{\pi (2{n_y} + 1 - {{\widetilde N}_y})}}{{2{{\widetilde N}_y}}},{n_y} = [(m - {n_x}({\widetilde N_y}{\widetilde N_z}))/{\widetilde N_z}],\\
			{\widetilde \gamma _m} &= \frac{{\pi (2{n_z} + 1 - {{\widetilde N}_z})}}{{2{{\widetilde N}_z}}},{n_z} = m - {n_x}({\widetilde N_y}{\widetilde N_z}) - {n_y}{\widetilde N_z},
		\end{align}
	\end{subequations}
	where $[\cdot ]$ denotes the greatest integer less than or equal to its argument. We then seek the best cuboid center that maximizes the objective function $G_{\min}$, denoted as
	\begin{equation}
		{m^ * } = \arg \mathop {\max }\limits_{1 \le m \le {{\widetilde N}_{tot}}} {G_{\min }}({\mathbf{r} _m}). 
	\end{equation}\par
	Once the best cuboid center $\mathbf{r}_{m^*}$ is identified, a finer-grained search is conducted within the $m^*$-th cuboid by further discretizing it into more sampling points. Let ${\widetilde N_{x,{m^ * }}}$, ${\widetilde N_{y,{m^ * }}}$ and ${\widetilde N_{z,{m^ * }}}$ denote the number of sampling angles for $\alpha$, $\beta$, and $\gamma$ within the $m^*$-th cuboid, respectively. As such, there are ${\widetilde N_{tot,{m^ * }}} = {\widetilde N_{x,{m^ * }}} \times {\widetilde N_{y,{m^ * }}} \times {\widetilde N_{z,{m^ * }}}$, sampling points in the $m^*$-th cuboid, and we denote by $\mathbf{r}_{m^*}^{(n)}$ the coordinate of the $n$-th sampling point in it. Then, the optimized rotational angles can be obtained as 
	\begin{equation}
		{\mathbf{r}^{\star}=\mathbf{r}_{m^*}^{(n^ \star )},{n^ \star } = \arg \mathop {\max }\limits_{1 \le n \le {{\widetilde N}_{tot,{m^ * }}}} {G_{\min }}({\mathbf{r}}_{m^*}^{(n)}}).\label{33}
	\end{equation}\par
	\subsection{GS for Solution Improvement}
   	Although the above hybrid search is generally effective in solving (P2.3), it may induce performance loss during the coarse-grained search. To address this issue, we propose improving its solution quality by further performing a GS process. The core idea of the GS is to explore solutions in the vicinity of the one obtained by the hybrid search, or to randomly jump to more distant solutions with significantly different rotation angles.\par
	Mathematically, let $\mathcal{S}$ denote the feasible set of the rotational angle vector $\mathbf{r}$. To handle the continuity of {$\cal S$}, we discretize the feasible space $[0,2\pi]$ of each rotational angle (i.e., $\alpha$, $\beta$, and $\gamma$) into several sampling angles. Let ${\Delta}$ denote the spacing between any two adjacent sampling angles; hence, we have $\lvert {\cal S} \vert = (\frac{2\pi}{\Delta})^3$.\par
	Assume that each GS phase consists of $T$ iterations and consider the $t$-th iteration. Let $\mathbf{r}^{(t-1)}$ represent the optimized rotational angle vector in the ($t-1$)-th GS  iteration, with $\mathbf{r}^{(t-1)} = [\alpha^{(t-1)}, \beta^{(t-1)}, \gamma^{(t-1)}]^T$. Define $\epsilon(t-1)=\{\mathbf{r}^{(i)}\}_{i=1}^{t-1}$ the set of optimized solutions by the GS until its $(t-1)$-th iteration, with $\epsilon(0)=\mathbf{r}^{\star}$ given in (\ref{33}). In each GS iteration, assume that $I$ candidate solutions are generated, with $I \ll \lvert {\cal S} \rvert$ 	
	and denote by $\mathbf{r}_i^{(t)}$ the $i$-th candidate solution in the $t$-th GS iteration, $i = 1, 2, ..., I$. These candidate solutions are drawn from two sets, respectively denoted as $\mathcal{B}(t)$ and $\mathcal{D}(t)$. The first set, $\mathcal{B}(t)$, contains nearby solutions  around $\mathbf{r}^{(t-1)}$ by appending each dimension with certain variations. Let $K$ denote the maximum variation (normalized by $\Delta$). As such, we can obtain the following $6K$ adjacent solutions around $\mathbf{r}^{(t-1)}$, i.e.,
	\begin{subequations}		
\begin{align}
	{\bf{r}}_{{k,1}}^{(t)} = {[\alpha^{(t - 1)} + k\Delta,\beta^{(t - 1)},\gamma^{(t - 1)}]^T},\\
	{\bf{r}}_{{-k,2}}^{(t)} = {[\alpha^{(t - 1)} - k\Delta,\beta^{(t - 1)},\gamma^{(t - 1)}]^T},\\
	{\bf{r}}_{{k,3}}^{(t)} = {[\alpha^{(t - 1)},\beta^{(t - 1)} + k\Delta,\gamma^{(t - 1)}]^T},\\
	{\bf{r}}_{{-k,4}}^{(t)} = {[\alpha^{(t - 1)},\beta^{(t - 1)} - k\Delta,\gamma^{(t - 1)}]^T},\\
	{\bf{r}}_{{k,5}}^{(t)} = {[\alpha^{(t - 1)},\beta^{(t - 1)},\gamma^{(t - 1)} + k\Delta]^T},\\
	{\bf{r}}_{{-k,6}}^{(t)} = {[\alpha^{(t - 1)},\beta^{(t - 1)},\gamma^{(t - 1)} - k\Delta]^T},
\end{align}
	\end{subequations}
	with $k \in {\cal K} = \{1,2,... ,K\}$. The second set, $\mathcal{D}(t)$, consists of ($I-6K$) randomly selected solutions from $\mathcal{S} \setminus \mathcal{B}$. By this means, we can obtain the $I$ candidate solutions of the rotational angle vector.
 	Next, each candidate solution in ${\cal B}(t) \cup {\cal D}(t)$ is assigned with a probability, based on which it may be selected as $\mathbf{r}^{(t)}$. The selection probability is given by
	\begin{equation}
		\begin{split}
			P_i^{(t)} &= \Pr \{ \mathbf{r}^{(t)} = {\mathbf{r}}^{(t)}_i|\mathbf{r}^{(t - 1)} \} \\
			&= \frac{{{e^{\mu {G_{\min }}({\mathbf{r}}^{(t)}_i)}}}}{{\sum\nolimits_{{{\mathbf{r}}}^{(t)}_0 \in {\cal B}(t) \cup {\cal D}(t)} {{e^{\mu {G_{\min }}({\mathbf{r}}^{(t)}_0)}}} }},i=1,2,...,I, \label{75}
		\end{split}
	\end{equation}
	where ${G_{\min }}(\mathbf{r}) = \mathop {\min }\limits_{\vartheta  \in {\cal D} \times {\cal F}} G(\mathbf{r},\vartheta )$, and $\mu \geq 0$ is a pre-defined scaling parameter. To determine $\mathbf{r}^{(t)}$ based on (\ref{75}), we randomly generate a float (denoted as $p_t$) between $0$ and $1$. Then, we update the solution as follows: 
	\begin{equation}
		\mathbf{r}^{(t)} = {\mathbf{r}^{(t)}_{i ^* }}, \label{76}
	\end{equation}
	where $i^*$ is the index that satisfies $\sum\nolimits_{i = 1}^{i ^*  - 1} {P_i^{(t)}}  < {p_t} \le \sum\nolimits_{i = 1}^{i ^* } {P_i^{(t)}}$.
	
	Followed by this, we update $\epsilon(t)=\epsilon(t-1) \cup \{\mathbf{r}^{(t)}\}$ and proceed to the $(t+1)$-th GS iteration. The GS process continues until the iteration number $t$ reaches the predefined maximum number of iterations, i.e., $T$. At the end of the GS, the solution that yields the best performance is selected from the set of solutions visited to replace $\mathbf{r}^{\star}$, which is given by
	\begin{equation}
		\mathbf{r}^{\star} = \arg \mathop {\max }\limits_{\mathbf{r} \in \epsilon (T)} {G_{\min }}(\mathbf{r}).
	\end{equation}
	Notably, as $\epsilon(T)$ includes the optimized solution by the hybrid search in (\ref{33}), the GS phase must yield an objective value of (P2) no worse than the latter. \vspace{-6pt}
	
	\subsection{Convergence and Complexity Analysis}
	The proposed AO algorithm optimizes the analog beamforming vector $\boldsymbol{\omega}$, the antenna position vector $\mathbf{p}$, and the rotation vector $\mathbf{r}$ in an alternate manner. Since each subproblem is solved with monotonic convergence, the convergence of the overall AO algorithm can be ensured. The main procedures of the AO algorithm are summarized in Algorithm 1.
	
	\begin{algorithm}
		\caption{Proposed AO algorithm to solve (P1)}
		\label{alg:P2}
		\begin{algorithmic}[1] 
			\State \textbf{Input:} $N$, $\boldsymbol{\omega}^{(0)}$, $\mathbf{p}^{(0)}$, $\mathbf{r}^{(0)}$, $T$.
			\State Initialization: $j \leftarrow 1$.
			\While{AO convergence is not reached}
			\State Initialize $i \leftarrow 0$ and update $\mathbf{W}^{(0)} = \boldsymbol{\omega}^{(j-1)} (\boldsymbol{\omega}^{(j-1)})^H$, $\mathbf{r} = \mathbf{r}^{(j-1)}$, and $\mathbf{p} = \mathbf{p}^{(j-1)}$.
			\While{SCA convergence for $\mathbf{W}$ is not reached}
			\State Obtain $\mathbf{W}^{(i+1)}$ by solving problem (P2.2).
			\State Update $i \leftarrow i+1$.
			\State Obtain $\boldsymbol{\omega}^{(j)}$ based on the SVD of $\mathbf{W}^{(i)}$.
			\State Initialize $i \leftarrow 0$ and update $\boldsymbol{\omega} = \boldsymbol{\omega}^{(j)}$, $\mathbf{r} = \mathbf{r}^{(j-1)}$, and $\mathbf{p}^{(0)} = \mathbf{p}^{(j-1)}$.
			\EndWhile
			\While{SCA convergence for $\mathbf{p}$ is not reached}
			\State Obtain $\mathbf{p}^{(i+1)}$ by solving problem (24).
			\State Update $i \leftarrow i+1$.
			\State Update $\mathbf{p} = \mathbf{p}^{(j-1)}$.
			\EndWhile
			\State Update $\mathbf{r}^{\star}=\mathbf{r}_{m^*}^{(n^{\star})}$ based on (\ref{33}).
			\State Initialize $t \leftarrow 1$, $\mathbf{r}^{(t)}=\mathbf{r}^{\star}$.
			\While{$t \le T$}
			\State Generate $\mathcal{B}(t)$ and $\mathcal{D}(t)$ and update $\mathbf{r}^{(t)}$ based on (\ref{76}).
			\State Update $\epsilon(t) = \epsilon(t-1) \cup \{\mathbf{r}^{(t)}\}$.
			\State Update $t \leftarrow t+1$.
			\EndWhile
			\State Update ${{\mathbf{r}}^{(t)}}$ as $\mathbf{r}^{(t^*)}$.
			\State Update $j \leftarrow j+1$.
			\EndWhile
			\State Update $\mathbf{p}=\mathbf{p}^{(j)}$, $\boldsymbol{\omega} = \boldsymbol{\omega}^{(j)}$ and $\mathbf{r} = \mathbf{r}^{(t)}$.
			\State \textbf{Output:} $\mathbf{r}$, $\mathbf{p}$, and $\boldsymbol{\omega}$.
		\end{algorithmic}
	\end{algorithm}
		
	Next, we analyze the computational complexity of the proposed AO algorithm as follows. The complexity is primarily dominated by the updates for $\boldsymbol{\omega}$, $\mathbf{p}$, and $\mathbf{r}$ in each AO iteration. The complexity older of solving the SCA problem for $\boldsymbol{\omega}$ using the interior-point method is $\mathcal{O}(\sqrt{N} (N^{2} + L_{tot}))$, where $L_{tot}$ is the total number of sampling points in the spatial-frequency domain. Similarly, the complexity order of  optimizing the antenna positions via SCA is on the order of $\mathcal{O}(\sqrt{N} (4N^{2} + L_{tot}))$. As for the hybrid search and the GS process, it can be shown that they yield a linear complexity order given by $\mathcal{O}((\widetilde N_{tot} + \widetilde N_{tot,m^*} + TI) N L_{tot})$. It thus follows that the AO algorithm yields a polynomial complexity order, which is tolerable for our considered problem relying on statistical CSI only.\vspace{-6pt}
	
	\section{Numerical Results}
	In this section, we provide numerical results to evaluate the performance of our proposed algorithm for mitigating beam squint with a 6DMA array. The system bandwidth is set to $B = 0.1$ THz, and the carrier frequency is ${f_c} = 1\text{ THz}$. The number of antennas is $N = 9$, and the minimum inter-antenna spacing is set to half-wavelength spacing (i.e., $d_{\min}=\lambda/2$). The range of the elevation AoD is $\phi=[0,\pi/2]$, and that of the azimuth AoD is $\theta=[0,\pi/2]$. In the AO algorithm, the transmit beamforming is initialized using a strategy similar to that in \cite{xie}. The antenna positions are initialized as those of a UPA, and the antenna rotation angles are initialized as ${\bf r}^{(0)}=[0,0,0]^T$.\vspace{-6pt} 
        \begin{figure}[t]
    \centering
		\subfigure[Beam gain versus frequency with $\theta = 60^\circ$.]{\includegraphics[width=0.35\textwidth]{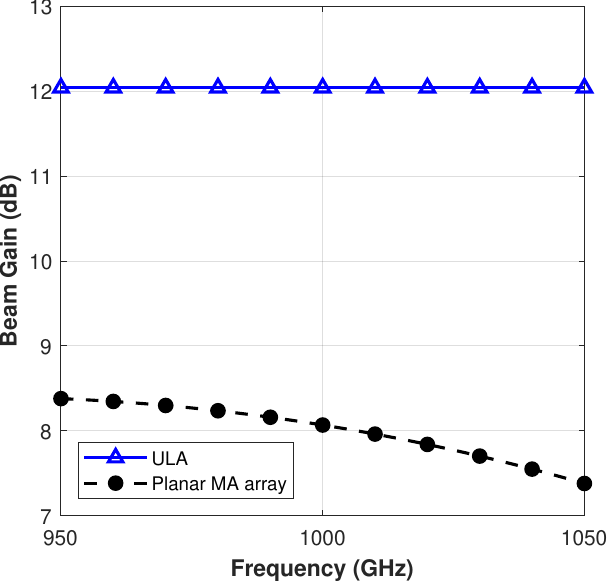}}
		\subfigure[Beam gain versus $\theta$ with $f_c = 1$ THz.]{\includegraphics[width=0.35\textwidth]{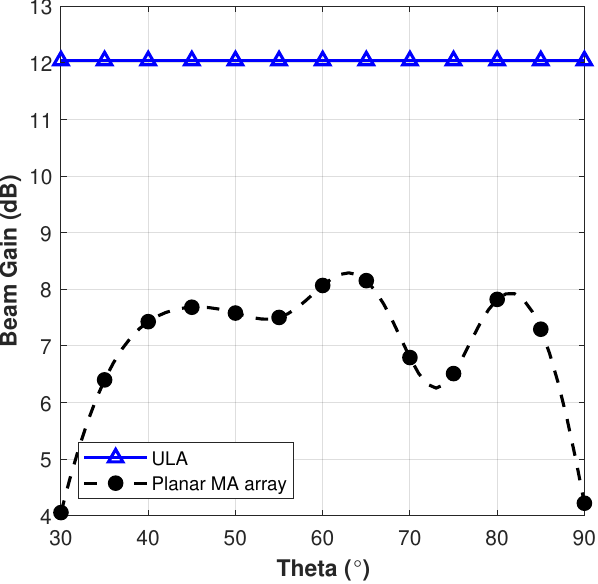}}
		\caption{Beam gains by ULA versus planar MA array for 1D angular coverage.}	
		\label{1dcoverage1}
        \vspace{-6pt}
	\end{figure}
		
	\subsection{1D Angular Coverage}
	To validate our analytical results presented in Section III, we compare the performance of a ULA and a planar MA array under a 1D wideband coverage scenario. The target coverage area is defined by a fixed azimuth angle $\phi=0^\circ$ and an elevation angle range $\theta \in [30^\circ, 90^\circ]$. Both arrays are equipped with $N=16$ antennas. For the planar MA array, the antenna movement region is assumed to be a square of $8\lambda \times 8\lambda$. To ensure its planar dimensions, we uniformly divide the antenna movement region into 16 subregions, each corresponding to an MA element. Each element is restricted to move within its associated subregion only. 
        \begin{figure}[t]
    \centering
		\subfigure[ULA]{\includegraphics[width=0.35\textwidth]{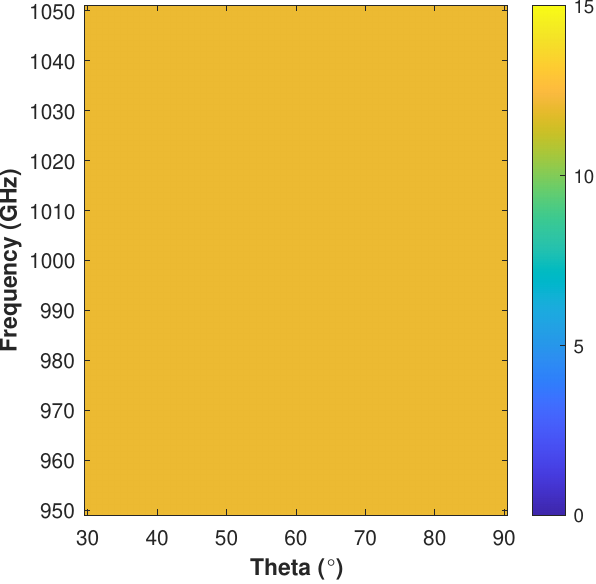}}
		\subfigure[Planar MA array]{\includegraphics[width=0.35\textwidth]{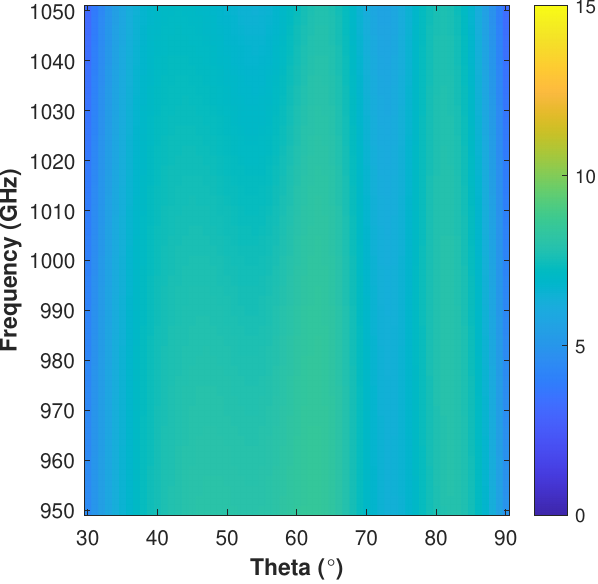}}
		\caption{Beam-gain distribution by ULA versus planar MA array for 1D angular coverage.}	
		\label{1dcoverage2}
        \vspace{-6pt}
	\end{figure}
    
    Figs.\,\ref{1dcoverage1}(a) and \ref{1dcoverage1}(b) plot the optimized beam gain versus the frequency range from 0.95 THz to 1.05 THz (with $\theta=60^\circ$) and versus the elevation angle (at the central frequency $f=1$ THz), respectively. In addition, we also plot the optimized beam gains by the ULA and the planar MA array across the spatial and frequency domains in Figs.\,\ref{1dcoverage2}(a) and Fig.\,\ref{1dcoverage2}(b), respectively. It is observed from  Fig.\,\ref{1dcoverage1}(a) that the ULA achieves a constant and frequency-independent beam gain of $12.04$ dB (which corresponds to the full array gain $10\log_{10}(16)$) across the entire bandwidth and angular range. Moreover, the ULA is observed to achieve a full beam gain of 16 in Fig.\,\ref{1dcoverage2}(a) as well. The above observations corroborate our analytical results presented in Theorem 1, i.e., rotating a ULA can eliminate beam-squint effects for 1D angular coverage, without requiring complex digital precoding.
	
	In contrast, for the planar MA array, even with joint optimization of wideband beamforming, antenna positions, and array rotation, the beam gain exhibits significant degradation and fluctuations, as observed from Figs.\,\ref{1dcoverage1}(a) and \ref{1dcoverage1}(b). In particular, the beam gain achieved by the planar MA array is approximately $4$-$8$ dB lower than that of the ULA and drops to around 4 dB at $\theta = 30^\circ$ in Fig.\,\ref{1dcoverage1}(b). Moreover, in Fig.\,\ref{1dcoverage2}(b), the beam-gain distribution of the planar MA array exhibits pronounced vertical ``stripes", indicating substantial variations across the frequency–angle domain. These observations validate our discussion at the end of Section \ref{special}.\vspace{-6pt}
    \begin{figure*}[t]
		\centering
		\subfigure[Narrowband beamforming with FPA.]{\includegraphics[width=0.32\textwidth]{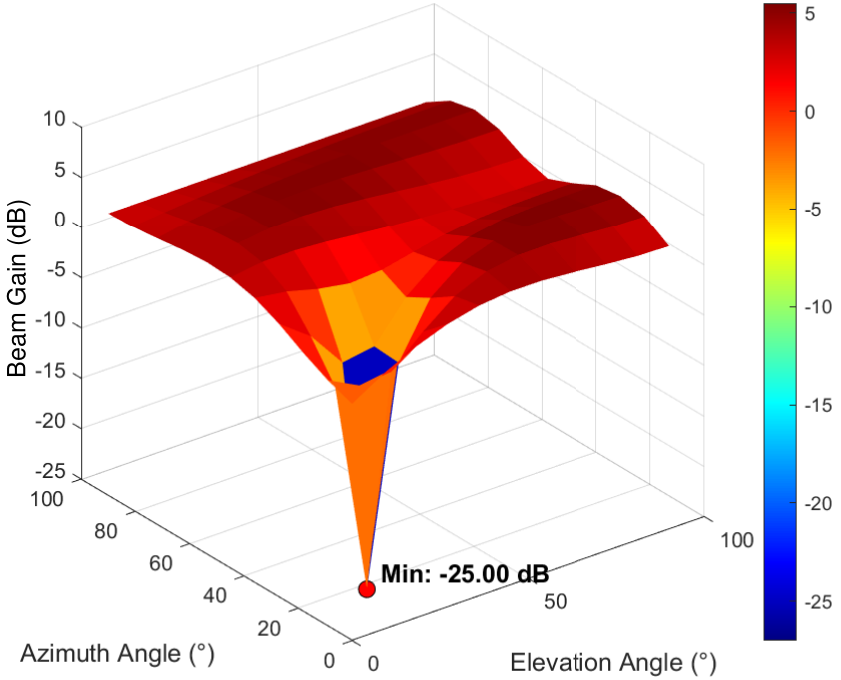}\label{heatmap1}}
		\subfigure[Wideband beamforming with FPA.]{\includegraphics[width=0.32\textwidth]{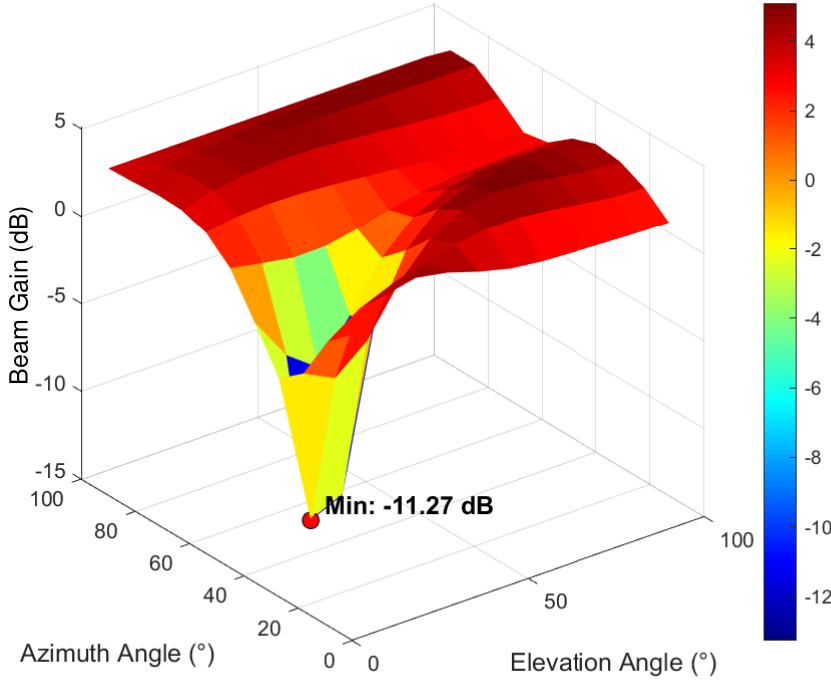}\label{heatmap2}} 
		\subfigure[Wideband beamforming with antenna movement only.]{\includegraphics[width=0.32\textwidth]{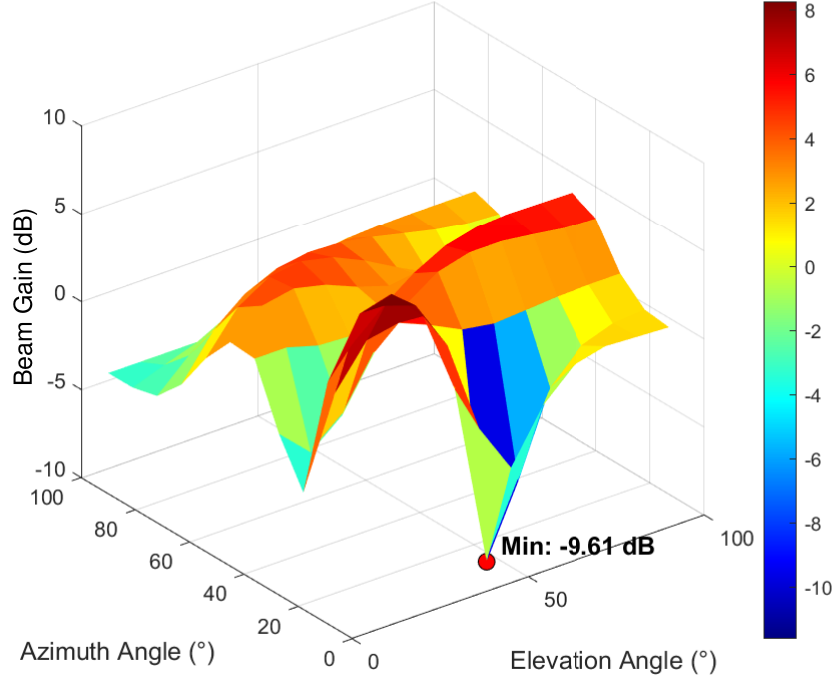}\label{heatmap3}}\\
		\subfigure[Wideband beamforming with antenna rotation only.]{\includegraphics[width=0.32\textwidth]{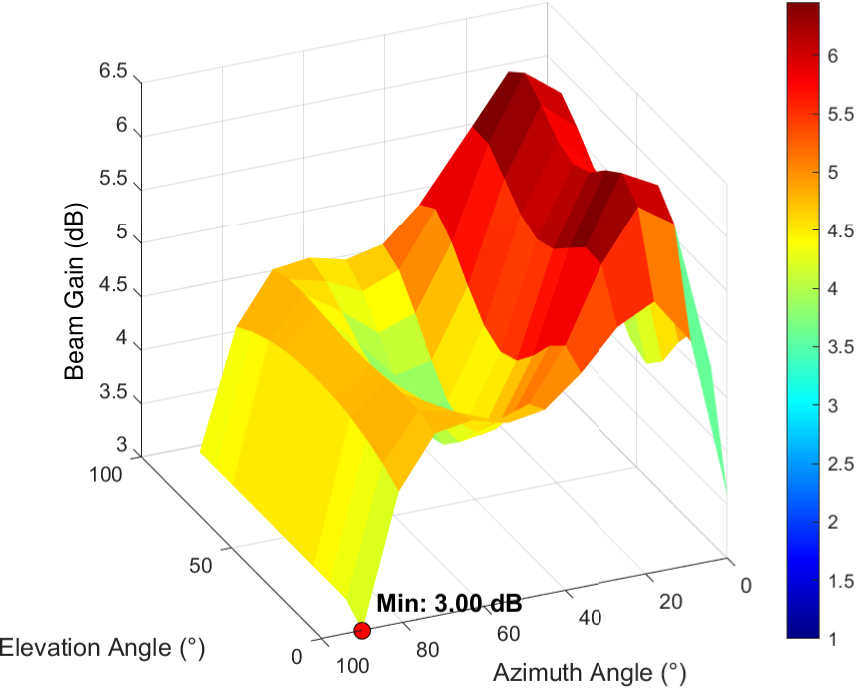}\label{heatmap4}}
		\subfigure[Proposed scheme with 6DMA.]{\includegraphics[width=0.32\textwidth]{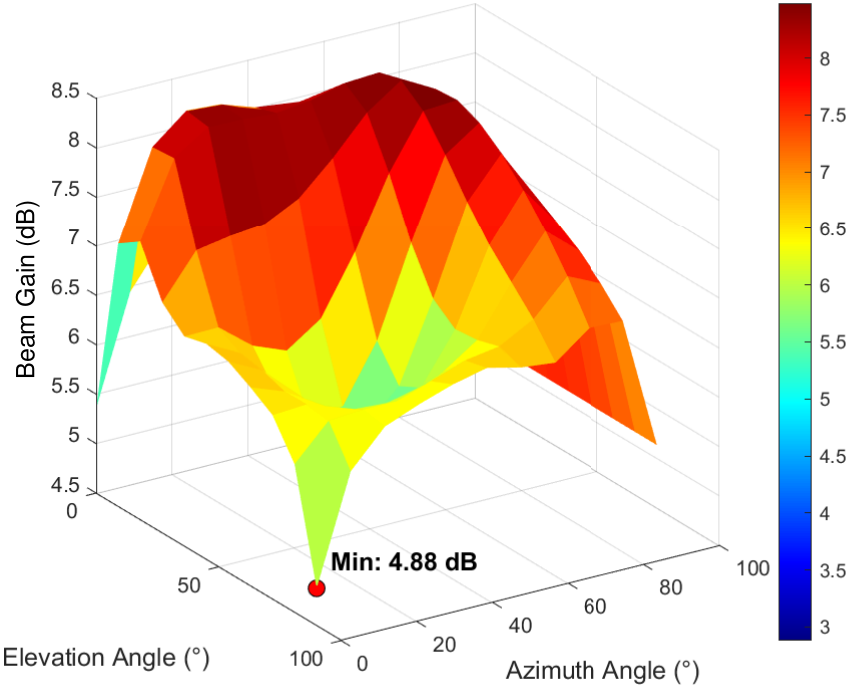}\label{heatmap5}} 
		\subfigure[Linear MA array.]{\includegraphics[width=0.32\textwidth]{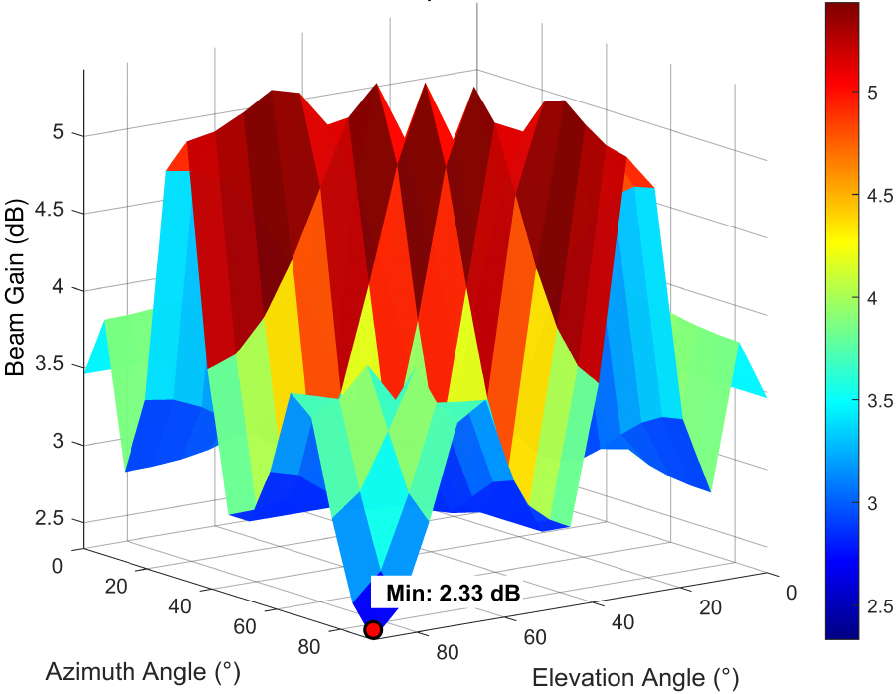}\label{heatmap6}}
		\caption{Beam-gain distribution of different schemes across the 2D angular domain at the central frequency.}	
		\label{heatmaps_comparison}
        \vspace{-6pt}
	\end{figure*}
	
    \subsection{2D Angular Coverage}
	For the more general scenario of 2D angular coverage, we first plot in Fig.\,\ref{heatmaps_comparison} the beam-gain distributions over the 2D angular domain at the central frequency by the following benchmarks:
    \begin{itemize}
        \item {\bf Narrowband beamforming with FPA} (Benchmark 1): The antennas are fixed without rotation and movement. The transmit beamforming is obtained by solving (P1) with $B=0$, which corresponds to narrowband beamforming.
        \item {\bf Wideband beamforming with FPA} (Benchmark 2): The transmit beamforming $\boldsymbol{\omega}$ is optimized based on the SCA algorithm presented in Section IV-A, without optimizing antenna rotation or movement.
        \item {\bf Wideband beamforming with antenna movement only} (Benchmark 3): The transmit beamforming $\boldsymbol{\omega}$ and antenna position $\mathbf{p}$ are optimized based on the SCA algorithms presented in Sections IV-A and IV-B, respectively.
        \item {\bf Wideband beamforming with antenna rotation only} (Benchmark 4): The transmit beamforming $\boldsymbol{\omega}$ and antenna rotation $\mathbf{r}$ are optimized based on the SCA and hybrid search algorithms presented in Sections IV-A and IV-C, respectively.
        \item {\bf Linear MA array} (Benchmark 5): The antennas are restricted to move along a linear array, with the transmit beamforming, antenna positions, and array rotation optimized similarly as in Section IV.
    \end{itemize}
    
    It is observed from Fig.\,\ref{heatmap1} that Benchmark 1 exhibits the most severe beam-gain fluctuations and frequency-spatial nulls due to the beam-squint effect, with the minimum beam gain dropping to as low as $-25.0$ dB. Compared with Benchmark 1, Fig.\,\ref{heatmap2} shows that Benchmark 2 can increase the minimum beam gain to $-11.3$ dB, indicating that optimizing transmit beamforming helps improve the overall performance by $13.7$ dB. By further incorporating antenna position/rotation optimization, Figs.\,\ref{heatmap3} and \ref{heatmap4} show that the minimum beam gain increases to $-9.6$ dB and $3.0$ dB, respectively. This suggests that antenna rotation optimization generally offers more significant advantages than position optimization for wideband wide-beam coverage. Fig. \ref{heatmap5} shows that by using the proposed AO algorithm that combines antenna position and rotation optimization, the minimum beam gain increases to $4.88$ dB, resulting in a more uniform beam-gain distribution compared with other benchmark schemes. Finally, Fig. \ref{heatmap6} shows that a linear MA array achieves a performance approximately 2.5 dB lower than our proposed scheme due to its geometric limitations. However, it still outperforms Benchmarks 1–3 and approaches the performance of Benchmark 4.
	
	Next, we plot the wideband beam gain versus bandwidth $B$ in Fig.\,\ref{2D MA_Benchmark1}. It is observed from Fig. \ref{2D MA_Benchmark1} that the wideband beam gains by all considered schemes decrease with the bandwidth $B$ due to the more severe beam-squint effects. Particularly, Benchmark 1 achieves the lowest beam gain among all schemes. Benchmark 2 raises the performance of Benchmark 1 by around $16$ dB. It is also observed that Benchmark 4 with antenna rotation optimization consistently yields a higher wideband beam gain than Benchmark 3 with antenna position optimization. This can be attributed to the fact that antenna rotation better aligns the array's boresight with the target region, thereby effectively mitigating the angle-dependent beam squint; whereas position optimization primarily mitigates beam squint through phase compensation. Moreover, the proposed 6DMA achieves the highest wideband beam gain, improving the performance of Benchmark 4 by approximately 3 dB. Its performance only experiences a slight decrease as $B$ increases, demonstrating that the proposed scheme effectively accommodates both wide frequency and angular ranges. Benchmark 5 is observed to achieve a comparable performance to Benchmark 4, which suggests that using a linear MA array can provide a cost-effective solution for wideband, wide-beam coverage.
	\begin{figure}[t]
		\centerline{\includegraphics[width=0.4\textwidth]{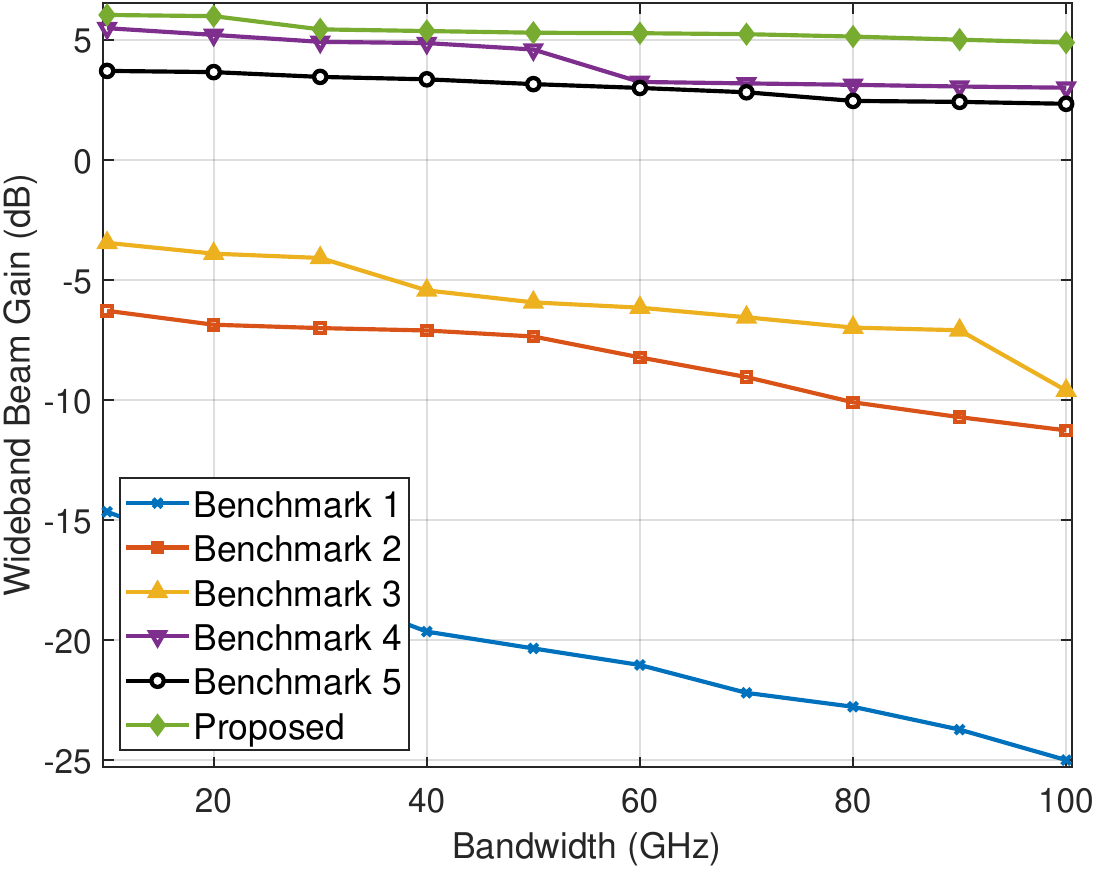}}
		\caption{Wideband beam gain versus bandwidth for different schemes.}
		\label{2D MA_Benchmark1}
        \vspace{-6pt}
	\end{figure} 	
	
	\begin{figure}[t]
		\centerline{\includegraphics[width=0.4\textwidth]{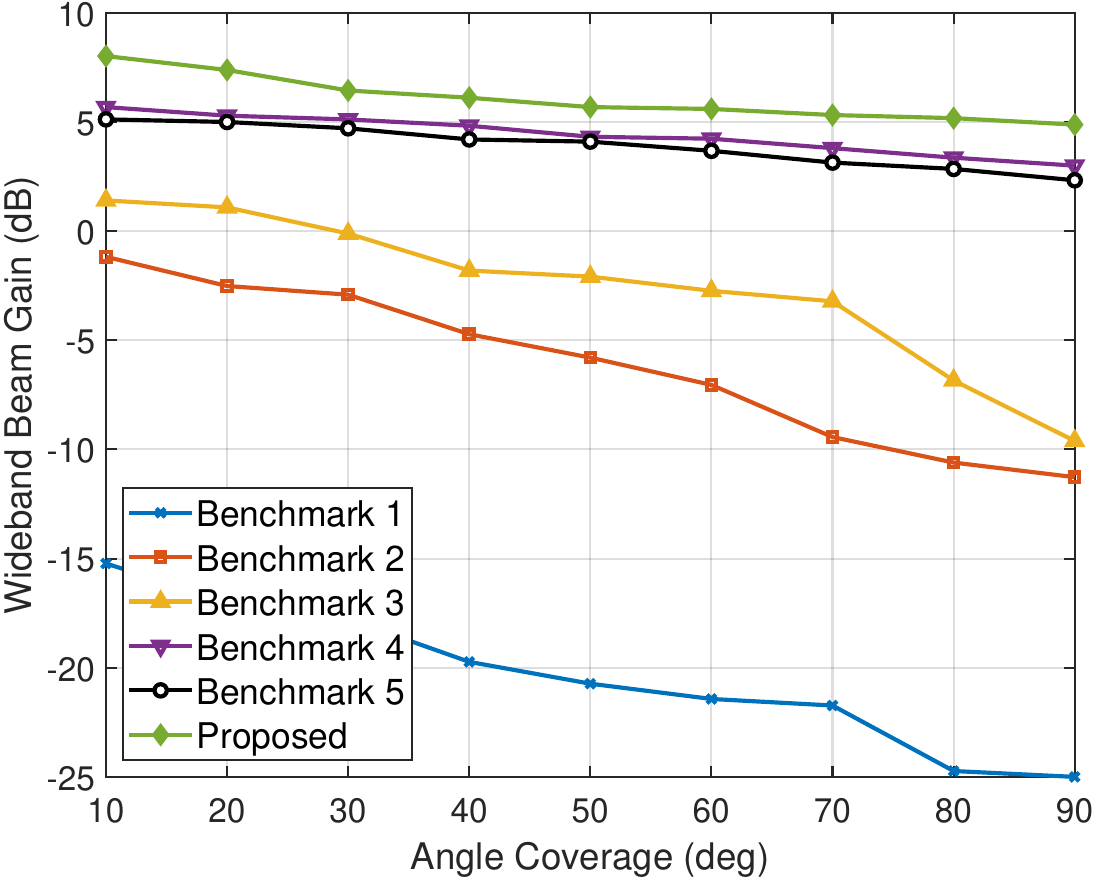}}
		\caption{Wideband beam gain versus width of azimuth AoD for different schemes.}
		\label{2D MA_Benchmark2}
        \vspace{-6pt}
	\end{figure} 	
	Fig.\,\ref{2D MA_Benchmark2} plots the wideband beam gains by different schemes versus the width of azimuth AoD (i.e., $\phi_{\max}-\phi_{\min}$), with elevation AoD $\theta \in [0,\pi/2]$ and bandwidth $B=100$ GHz. The main observations are similar to those made from Fig.\,\ref{2D MA_Benchmark1}. Specifically, Benchmark 1 still achieves the worst performance among all considered schemes. Even for a narrow width of $\phi_{\max}-\phi_{\min}=10^\circ$, the proposed scheme and Benchmarks 2-5 can enhance the wideband beam gain by approximately 23 dB, 11 dB, 13 dB, 20.5 dB, and 20 dB, respectively. Moreover, the wideband beam gains of all considered schemes decrease as $\phi_{\max}-\phi_{\min}$ increases due to the larger angular region that needs to be covered. Nonetheless, the performance degradation of the proposed scheme remains small with increasing $\phi_{\max}-\phi_{\min}$, especially when compared to Benchmarks 1-3.
    
	\begin{figure}[t]
		\centerline{\includegraphics[width=0.38\textwidth]{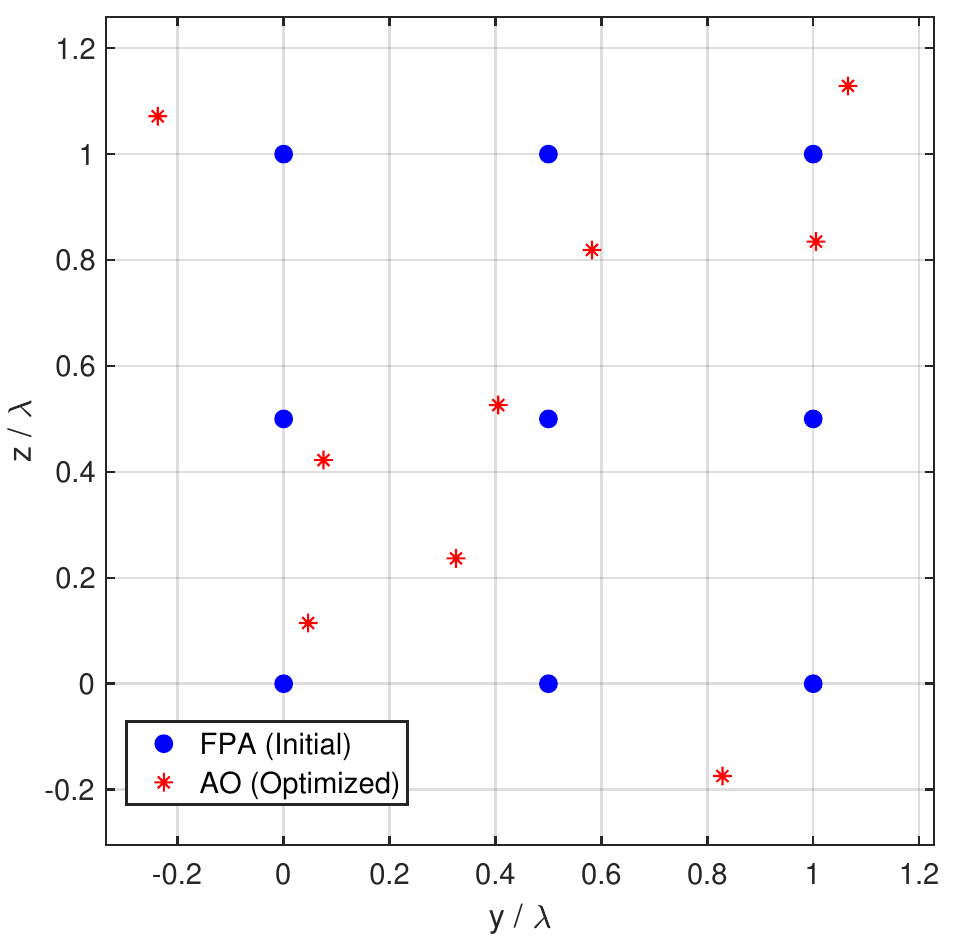}}
		\caption{Optimized 2D antenna positions by AO.}
		\label{position}
	\end{figure} 
	\begin{figure}[t]
		\centerline{\includegraphics[width=0.45\textwidth]{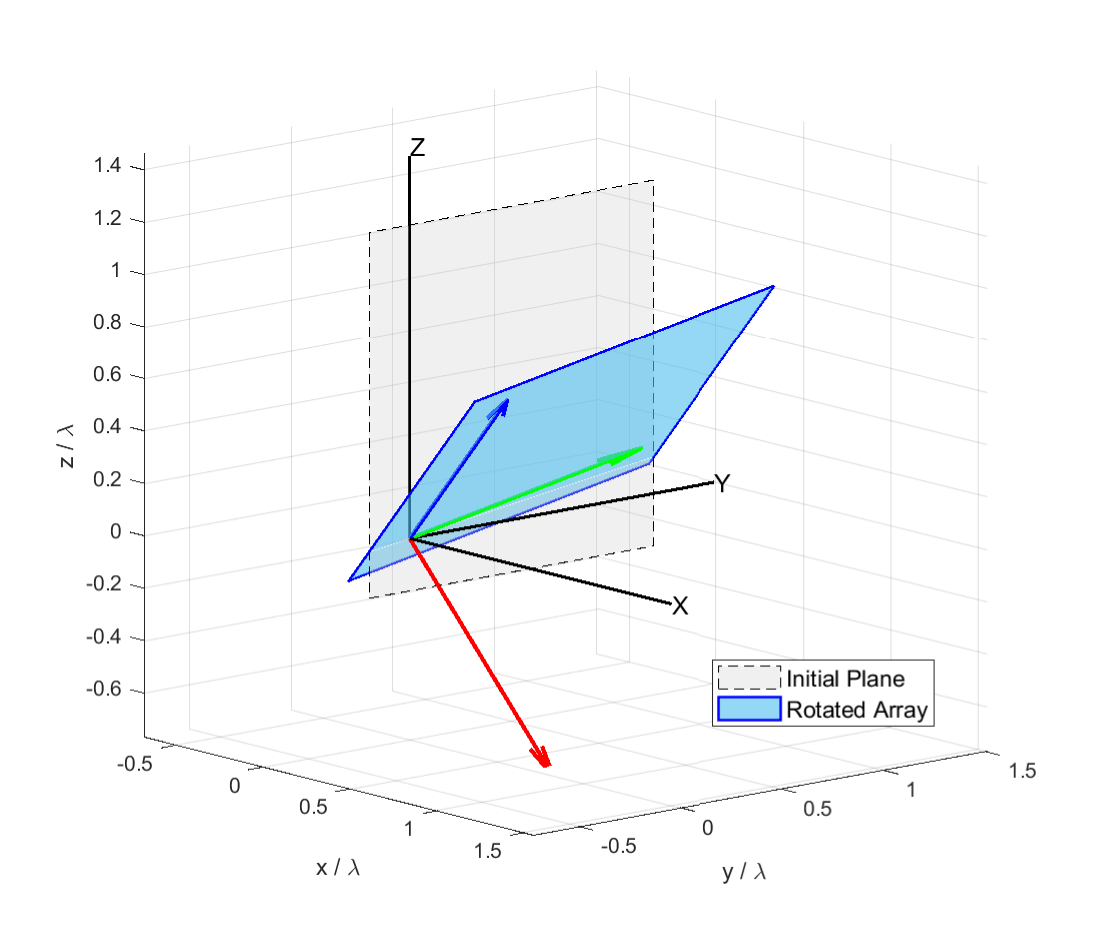}}
		\caption{Optimized 3D rotational angles by AO.}
		\label{rotation}
        \vspace{-6pt}
	\end{figure} 
	Finally, we plot in Figs.\,\ref{position} and \ref{rotation} the optimized antenna positions and rotation, respectively. In Fig.\,\ref{rotation}, the black coordinate system represents the original plane of the antenna array, with the $Y$-$O$-$Z$ plane parallel to the array and the $x$-axis perpendicular to it. The blue, green, and red axes correspond to the $x$-, $y$-, and $z$-axes after the array is rotated, respectively. Specifically, the optimized rotation angles are yielded as $\alpha = 7.02^\circ$, $\beta = 48.04^\circ$, and $\gamma = -7.73^\circ$ around the $x$-, $y$-, and $z$-axes, respectively. It is observed from Fig.\,\ref{position} that the AO-optimized antenna geometry exhibits greater sparsity and a larger effective aperture than the FPA, which facilitates wideband wide-beam coverage. Moreover, the optimized antenna rotation shown in Fig. \ref{rotation} implies that 3D rotation is needed to optimally balance the power distribution in both angular and frequency domains.

	\section{Conclusion}
	In this paper, we investigated the application of 6DMA for THz beam squint mitigation in wideband wide-beam coverage. We aimed to jointly optimize the transmit beamforming, antenna positions, and rotation angles to maximize the minimum beam gain over the desired spatial-frequency domain. To gain insights, we derived a closed-form optimal solution to this problem in the special case of 1D angular coverage. Our analytical results showed that 1D array rotation can achieve global optimality and eliminate beam-squint effects in this case. For other general cases, we developed an AO algorithm that combines SCA and GS to obtain a high-quality suboptimal solution. Numerical results demonstrated that our proposed scheme significantly outperforms other baseline schemes, confirming its efficacy in mitigating beam-squint effects. Furthermore, it was shown that antenna rotation offers more significant advantages for boosting the wideband wide-beam coverage performance than antenna repositioning. This paper can be extended in several directions for future work. For example, the proposed 6DMA scheme could be generalized to the near-field scenario, where the spherical wavefront introduces significant range-dependent beam-squint effects in addition to the angle-dependent effects explored in this paper. Furthermore, developing robust 6DMA schemes that account for hardware impairments in THz phased arrays will be crucial for their practical deployment in real-world scenarios. Last but not least, more efficient optimization algorithms can be developed to address the considered problem, which jointly involves spatial and frequency domains as well as antenna position and rotation optimization, such as randomized optimization techniques \cite{wang_twc_max-min}.\vspace{-6pt}
	
	
	\appendices
	
	\section{ULA for 2D Coverage}
In this section, we prove that it is impossible to eliminate beam-squint effects for all directions in a 2D region using array rotation only. First, under 2D angular coverage, the phase term in \eqref{an} becomes
\begin{equation}
	\mathbf{r}_1^T \mathbf{v}(\theta,\phi) = \left[ (r_{11}\cos\phi + r_{21}\sin\phi)\cos\theta + r_{31}\sin\theta \right], \label{an-2D}
\end{equation}
where $r_{11} = {c_\beta }{c_\gamma }$, $r_{21} = {s_\alpha}{s_\beta}{c_\gamma} + {c_\alpha}{s_\gamma}$, and $r_{31} = {s_\alpha}{s_\gamma} - {c_\gamma}{s_\beta}{c_\gamma}$, satisfying the unit-norm constraint $r_{11}^2 + r_{21}^2 + r_{31}^2 = 1$. 

Let $K_{0} = \sqrt{r_{11}^2 + r_{21}^2}$ and $\delta = \arctan(r_{21}/r_{11})$. Eq. \eqref{an-2D} can be expressed as
\begin{equation}
	\mathbf{r}_1^T \mathbf{v}(\theta,\phi) = \left[ {K_{0} \cos(\phi - \delta)} \cos\theta + {r_{31}} \sin\theta \right],
\end{equation}
where the two coefficients $K_0$ and $r_{31}$ only depend on the rotational angles, regardless of $\theta$ and $\phi$.

Notably, to achieve a ``zero-squint'' wideband beam gain across the entire 2D region, the coefficient term $\mathbf{r}_1^T \mathbf{v}(\theta,\phi)$ must be zero for all $\theta \in [\theta_{\min}, \theta_{\max}]$ and $\phi \in [\phi_{\min}, \phi_{\max}]$. 
To this end, it must hold that $K_0=r_{31}=0$. Considering $K_0=0$, we have
\begin{equation}
	K_0 = \sqrt{r_{11}^2 + r_{21}^2} = 0 \implies r_{11} = 0, \quad r_{21} = 0.
\end{equation}
As $r_{11}^2 + r_{21}^2 + r_{31}^2 = 1$, it follows that $r_3=\pm 1$. Substituting $r_3=\pm 1$ into (\ref{an-2D}), the phase term becomes
\begin{equation}
	\mathbf{r}_1^T \mathbf{v}(\theta,\phi) = \pm \frac{2\pi f}{c} x_n \sin\theta.
\end{equation}
It is observed that the phase term still depends on $\sin\theta$, which implies that beam-squint effects remain for ULA.

	\section{UPA for 1D Coverage}
The array response of the $n$-th antenna within a UPA can be expressed as
\begin{align}
	{a_n}(f,\theta ) 
	&= {e^{j\frac{{2\pi f}}{c}({\mathbf{v}^T(\theta)\mathbf{s}_1}{y_n} + \mathbf{v}^T(\theta){\mathbf{s}_2}{z_n})}}\nonumber\\
    &={e^{j\frac{{2\pi f}}{c}(g_1(\theta){y_n} + g_2(\theta){z_n})}},
\end{align}
where we have defined $g_1(\theta) = \mathbf{v}^T(\theta)\mathbf{s}_1$ and $g_2(\theta) = \mathbf{v}^T(\theta)\mathbf{s}_2$. 

To eliminate the beam-squint effects over the two dimensions of the UPA, it must hold that
\begin{equation}
	\left\{ {\begin{array}{*{20}{c}}
			{{g_1}(\theta ) = {{\bf{v}}^T}(\theta ){{\bf{s}}_1} = 0},\\
			{{g_2}(\theta ) = {{\bf{v}}^T}(\theta ){{\bf{s}}_2} = 0}.
	\end{array}} \right.
\end{equation}
Recall that
\begin{equation}
	\begin{split}
		\mathbf{v}(\theta) =& {[\cos \theta \cos {\phi _0},\cos \theta \sin {\phi _0},\sin \theta ]^T}
		\\	=& \cos\theta \mathbf{u}_{\text{planar}} + \sin\theta \mathbf{u}_{\text{vertical}},
	\end{split}
\end{equation}
where $\mathbf{u}_{\text{planar}} = [\cos\phi_0, \sin\phi_0, 0]^T$ and $\mathbf{u}_{\text{vertical}} = [0, 0, 1]^T$ are two constant and orthogonal unit vectors.
Then, $g_1(\theta)$ can be expressed as
\begin{equation}
	{{g_1}(\theta ) = \cos \theta {{\bf{u}}^T_{{\rm{planar}}}}{{\bf{s}}_1} + \sin \theta {{\bf{u}}^T_{{\rm{vertical}}}}{{\bf{s}}_1} = 0}. \label{g1_ap}
\end{equation}
To ensure that ${g_1}(\theta )=0, \forall \theta$, the coefficients of $\sin \theta$ and $\cos \theta$ must be zero, which implies that $\mathbf{s}_1$ must be orthogonal to both $\mathbf{u}_{\text{planar}}$ and $\mathbf{u}_{\text{vertical}}$. Notably, the direction vector orthogonal to two orthogonal vectors is unique as $\mathbf{e}_{\text{v}} = \mathbf{u}_{\text{planar}} \times \mathbf{u}_{\text{vertical}}$. Similarly, for $g_2(\theta )=0, \forall \theta$ to hold, the vector $\mathbf{s}_2$ must also be orthogonal to both $\mathbf{u}_{\text{planar}}$ and $\mathbf{u}_{\text{vertical}}$, implying that $\mathbf{s}_2$ must be collinear with $\mathbf{e}_{\text{v}}$ as well.

This leads to the conclusion that $\mathbf{s}_1$ and $\mathbf{s}_2$ are collinear. However, $\mathbf{s}_1$ and $\mathbf{s}_2$ are distinct columns of a rotation matrix $\mathbf{R}$ and thus must be orthogonal (i.e., $\mathbf{s}^T_1\mathbf{s}_2=0$). This results in a contradiction. Therefore, it is impossible for $g_1(\theta)$ and $g_2(\theta)$ to be zero simultaneously. The beam-squint effects cannot be eliminated via rotation.

	\bibliographystyle{IEEEtran}

\end{document}